\AtBeginDocument{%
	\providecommand\BibTeX{{%
			\normalfont B\kern-0.5em{\scshape i\kern-0.25em b}\kern-0.8em\TeX}}}

\documentclass[10pt,letterpaper]{article}
\usepackage[margin=1in]{geometry}
\usepackage{times}

\usepackage[round]{natbib}

\usepackage{graphicx}
\usepackage{float}

\usepackage{amsmath,amssymb,amsfonts}%
\usepackage{mathrsfs}%
\usepackage[title]{appendix}%
\usepackage{textcomp}%
\usepackage{manyfoot}%
\usepackage{algorithm}%
\usepackage{algorithmicx}%
\usepackage{algpseudocode}%
\usepackage{listings}%

\usepackage{booktabs}
\usepackage{enumitem}
\usepackage{dsfont}
\usepackage[labelformat=simple]{subcaption}

\usepackage{xspace}
\usepackage{amsthm}
\usepackage{multirow}
\theoremstyle{definition}
\newtheorem{definition}{Definition}[]
\usepackage{mathtools}
\usepackage{array}
\usepackage{lipsum}
\usepackage{xcolor}
\usepackage[compact]{titlesec}
\usepackage[T1]{fontenc}
\usepackage[utf8]{inputenc}
\usepackage{authblk}
\usepackage{url}
\usepackage{placeins}

\usepackage{adjustbox}
\usepackage{tabularx}
\usepackage[vlines]{tabularht}
\usepackage[figuresright]{rotating}
\newcommand{\ours}{\textsc{PEAR}\xspace}
\newcommand{\ourslong}{\bus{P}attern-based \bus{E}dge-weight \bus{A}ssignment on g\bus{R}aphs}

\usepackage{contour}
\usepackage[normalem]{ulem}
\contourlength{0.8pt}

\newcommand{\myuline}[1]{%
	\uline{\phantom{#1}}%
	\llap{\contour{white}{#1}}%
}

\newcommand{\smallsection}[1]{{\noindent {\bolden{\myuline{#1}}}}}

\newcommand{\bus}[1]{\textbf{\underline{\smash{#1}}}}

\newcommand{\overbar}[1]{\mkern 1.5mu\overline{\mkern-1.5mu#1\mkern-1.5mu}\mkern 1.5mu}

\definecolor{peace}{RGB}{228, 26, 28}
\definecolor{love}{RGB}{55, 126, 184}
\definecolor{joy}{RGB}{77, 175, 74}
\definecolor{kindness}{RGB}{152, 78, 163}

\newcommand{\natnum}{\mathbb{N}}
\newcommand{\setpar}[1]{\{#1\}}

\makeatletter
\newcommand{\AlignFootnote}[1]{%
	\ifmeasuring@
	\else
	\iffirstchoice@
	\footnote{#1}%
	\fi
	\fi}
\makeatother

\usepackage{bm}
\newcommand\bolden[1]{{\boldmath\bfseries#1}}

\newtheorem{theorem}{Theorem}
\newtheorem{remark}{Remark}%
\newtheorem{observation}{Observation}%

\newtheorem{problem}{Problem}%

\title{Interplay between Topology and Edge Weights in Real-World Graphs: Concepts, Patterns, and an Algorithm}

\author{Fanchen Bu\thanks{School of Electrical Engineering, KAIST, Daejeon, South Korea, boqvezen97@kaist.ac.kr}, \
Shinhwan Kang\thanks{Kim Jaechul Graduate School of AI, KAIST, Seoul, South Korea, shinhwan.kang@kaist.ac.kr}, \ and 
Kijung Shin\thanks{Kim Jaechul Graduate School of AI and School of Electrical Engineering, KAIST, Seoul, South Korea, kijungs@kaist.ac.kr}}

\date{}

\begin{document}

\maketitle

\begin{abstract}   
	What are the relations between the edge weights and the topology in real-world graphs?
	Given only the topology of a graph, how can we assign realistic weights to its edges based on the relations? 
	Several trials have been done for \textit{edge-weight prediction} where some unknown edge weights are predicted with most edge weights known.
	There are also existing works on generating both topology and edge weights of weighted graphs.
	Differently, we are interested in generating edge weights {that are realistic in a macroscopic scope}, merely from the topology, which is unexplored and challenging.
	To this end, we explore and exploit the patterns involving edge weights and topology in real-world graphs.
	Specifically, we divide each graph into \textit{layers} where each layer consists of the edges with weights at least
	a threshold.
	We observe consistent and surprising patterns appearing in multiple layers:
	the similarity between being adjacent and having high weights,
	and the nearly-linear growth of the fraction of edges having high weights with the number of common neighbors.
	We also observe a power-law pattern that connects the layers.
	Based on the observations, we propose \ours, an algorithm assigning realistic edge weights to a given topology.
	The algorithm relies on only \textit{two} parameters, preserves \textit{all} the observed patterns, and produces more realistic weights than the baseline methods with {more parameters}.
\end{abstract}

\maketitle

\section{Introduction}\label{sec:intro}
In weighted graphs, the edge weights reveal the heterogeneity of edges and enrich the information provided by the topology~\citep{newman2004analysis}.
In practice, weighted graphs have been widely used to model traffic~\citep{de2007structure}, biological interactions~\citep{aittokallio2006graph}, personal preference~\citep{liu2009personal}, etc.
The relation between topology and edge weights, therefore, attracts much attention.
A typical scenario where the two kinds of information are integrated is \textit{edge-weight prediction}~\citep{fu2018link, rotabi2017detecting}.
The target of edge-weight prediction is to predict the unknown edge weights using the given topological and edge-weight information, where usually most of the edge weights are given as the inputs.
Another related direction is to generate both the topology and the edge weights of weighted graphs~\citep{akoglu2008rtm, mcglohon2008weighted, yang2021hidden}.

However, not much has been explored about the relation between the pure topology and the edge weights in a graph, despite the importance of the relation.
In some previous trials, the problem of classifying edges into strong ones and weak ones  by assuming strong triadic closure~\citep{sintos2014using} (STC) is considered. 
The STC assumption forbids open triangles (also called triads or wedges) with two strong edges and aims to maximize the number of strong edges.
However, the diversity of edge weights is over-simplified in such a setting.
Moreover, it has been pointed out by~\cite{adriaens2020relaxing} that the STC assumption with a maximum number of strong edges is often far from reality.
The generalized version considered by~\cite{adriaens2020relaxing} still has room for improvement, especially w.r.t the empirical grounds, as shown in the experimental results.

Specifically, we study how the edge weights in real-world graphs are related to the topology \textit{in a macroscopic way}, which allows us to generate realistic edge weights when given an unweighted topology.
We would like to emphasize that we do not aim to assign edge weights with small errors w.r.t each individual edge.
We are motivated by the following practical applications:
\begin{itemize}
	\item \textbf{Edge weight anonymization.}
	In social networks, due to data privacy issues, sometimes only the binary connections are publicly accessible, while the detailed edge weights should not be publicized~\citep{steinhaeuser2008community, skarkala2012privacy}.
	Using the macroscopic patterns, we are able to generate realistic edge weights for a given topology and publicize the generated weighted graph to researchers and practitioners as a benchmark dataset, without revealing the true edge weights.
	\item \textbf{Anomaly detection.}
	In communication networks, the edge weights usually represent the frequency or intensity of the communication between the entities.
	Using the patterns observed on real-world graphs, we may detect anomalous edge weights that deviate from the patterns, and they may correspond to entities that have abnormally frequent or intensive communication~\citep{thottan2010anomaly, akoglu2015graph}. 
	\item \textbf{Community detection.}
	Community detection is a fundamental problem in network analysis~\citep{fortunato2010community}. 
	Edge weights are known to be helpful for community detection because they provide additional information about the strength and importance of connections between nodes~\citep{liu2014weighted, he2021statistical}, and thus assigning edge weights to unweighted graphs has the potential to enhance the performance of community detection algorithms~\citep{berry2011tolerating}.\footnote{{See Appendix~\ref{app:application_com_det} for some illustrative experiments, where we use edge weights generated by our proposed method to enhance the performance of a community detection method.}}
	\color{black}
\end{itemize}

We introduce and use a new tool called \textit{layers} to study weighted graphs in a hierarchical way, where each layer is a subgraph that consists of the edges with weights exceeding some threshold.
We examine eleven real-world graphs from five different domains and observe consistently strong correlations between the number of common neighbors (CNs) of an edge and the weight of the edge.
Although the information of CNs has been widely used in \textit{link prediction}~\citep{wang2015link} and used to indicate the significance of \textit{individual} edges~\citep{ahmad2020missing,cao2015grarep,zhu2016link}, to the best of our knowledge, we are the first to study the \textit{quantitative} patterns between the information of CNs and edge weights in a \textit{macroscopic} scope.
We observe consistent within-layer patterns in multiple layers:
(1) the nearly-linear growth of fraction of high-weight edges with the number of CNs,
(2) the relation between being adjacent and having high weights 
{(specifically, the relation between the fraction of high-weight edges and that of adjacent pairs with the same name number of CNs)}, and
across the layers, we observe a power-law correlation between the overall fraction of high-weight edges and the counterpart within the group of edges sharing no CNs.
Based on the observations, we propose \ours (\ourslong), an algorithm for assigning realistic edge weights to a given topology by preserving all the observed macroscopic patterns.
The proposed algorithm has only \textit{two} parameters.
On multiple real-world datasets, \ours outperforms the baseline methods using the same number of, or even more, parameters, producing more realistic edge weights in several different aspects w.r.t different macroscopic network statistics.

In short, our contributions are five-fold:
\begin{itemize}
	\item \textbf{New problem.}
	We introduce a new challenging problem:
	realistic assignment of edge weights \textit{merely} based on \textit{topology}.
	\item \textbf{New perspective.}
	We introduce the concept of \textit{layers}, which provides a new perspective to study weighted graphs.
	\item \textbf{Patterns.}
	We extensively study {eleven} real-world graphs and discover the various relations between topology and edge weights.
	\item \textbf{Algorithm.}
	We propose \ours, a weight-assignment algorithm based on the observed patterns. The algorithm has only \textit{two} parameters yet produces realistic edge weights to a given topology.
	\item \textbf{Experiments.}
	We evaluate \ours on real-world graphs.
	Without sophisticated fine-tuning, \ours overall outperforms the baseline methods with {more parameters,}
	producing more realistic edge-weights w.r.t node-degree and edge-CN distributions, average clustering coefficient, {and a graph distance measure computed by NetSimile~\citep{berlingerio2012netsimile}}. 
\end{itemize}

\smallsection{Roadmap.}
The remaining part of the paper is organized as follows.
In Section~\ref{sec:relwk}, we discuss related work.
In Section~\ref{sec:prelim},
we provide some preliminaries. 
In Section~\ref{sec:concepts}, we propose some new concepts.
In Section~\ref{sec:patterns},
we describe the patterns that we observe on real-world datasets.
In Section~\ref{sec:model}, we formulate our observations and, based on them, propose our algorithm, \ours.
In Section~\ref{sec:eval},
the empirical evaluation of \ours on real-world datasets is demonstrated.
In Section~\ref{sec:disc_concl}, we discuss some potential limitations of our work and future directions, and lastly, conclude the paper.

\smallsection{Reproducibility.}
The code and datasets are available online~\citep{onlineSuppl}.\footnote{\url{https://github.com/bokveizen/topology-edge-weight-interplay}}

\section{Related work}\label{sec:relwk}
\smallsection{Edge-weight prediction.}
In the early trials of edge-weight prediction~\citep{aicher2015learning,zhao2015prediction,zhu2016weight}, the problem is dealt with as a natural extension of the link prediction~\citep{martinez2016survey} problem.
Specifically, the proposed link-prediction algorithms assign scores to node pairs as the likelihood of edge existence, and the scores are also naturally used as the estimated edge weights.
More recently, \cite{fu2018link} use multiple topological features to predict the unknown edge weights in a supervised manner.
Specifically, they fit a regression model to the known edge weights with the features and use the fitted model to predict the unknown ones.
The main differences between the problem that we focus on in this work and the edge-weight prediction problem are:
(1) in the edge-weight prediction problem, most (e.g., 80\%~\citep{aicher2015learning} or 90\%~\citep{fu2018link,zhao2015prediction,zhu2016weight}) of the edge weights are assumed to be known and are given together with topology as the inputs, while we consider the scenarios where we only have access to the topology and we have \textit{none} known edge weights; and
(2) the target of the edge-weight prediction problem is to estimate the weights of individual edges in a \textit{microscopic} way, while we aim to generate realistic edge weights for a given topology preserving the \textit{macroscopic} patterns that we observe in real-world graphs.
Notably, there is another independent research problem that focuses on the edge-weight prediction of weighted signed graphs, which has essential differences from the research problem in this paper.
Specifically, the techniques proposed in a recent work studying that problem \citep{kumar2016edge} are specially designed for weighted signed graphs representing pairwise relations such as like/dislike and trust/distrust, which cannot be directly applied to the scenarios that we focus on where the edge weights represent the repetitions of the corresponding binary relations.

\smallsection{Weighted-graph generation.}
The other trials exploring the interplay between topology and edge weights, include the weighted-graph generation problem~\citep{akoglu2008rtm, mcglohon2008weighted,yang2021hidden}.
In those works, the authors specifically study the evolution of both topology and edge weights over time, and they propose algorithms that generate both topology and edge weights of weighted graphs.
Although some simple static patterns (e.g., power-law or geometric weight distributions) are also discussed in those works, the problem that we focus on, generating edge weights for a given topology, is essentially and technically different with the weighted-graph generation problem since in our problem the topology is given and thus fixed.

\smallsection{Strong triadic closure.}
The concept of strong triadic closure (STC) is first proposed by~\cite{sintos2014using}, where the authors consider the problem of classifying edges into strong ones and weak ones.
They define that a graph satisfies the STC property if there exists no open triangle with two strong edges,\footnote{Formally, the STC property requires that, for any three nodes $u$, $v$, and $w$, if both of the edges $(u, v)$ and $(u, w)$ are strong, then the edge $(v, w)$ must exist.}
and they assume that graphs often satisfy the STC property and have many strong edges.
Therefore, they specifically consider the problem of maximizing the number of strong edges while satisfying the STC property.
However, only two types of edge weights are considered in this problem, while real-world graphs often have a high diversity of edge weights (see, e.g., the datasets in Table~\ref{tab:datasets}).
Moreover, this optimization problem has both theoretical (it can have many optimal solutions) and practical (real-world graphs often do not have many strong edges) limitations, as pointed out by \cite{adriaens2020relaxing}.
Even though the above problem has been extended to edge weights of a wider range with other modifications~\citep{adriaens2020relaxing}, the extended version still has room for improvement w.r.t the empirical grounds, and the methods fail to predict the edge weights of real-world graphs accurately.
Specifically, the predicted edge weights have almost zero correlation with the ground truth on many datasets, as shown by~\cite{adriaens2020relaxing}.

To the best of our knowledge, we are the first to consider the problem of assigning realistic edge weights to a given topology by trying to preserve patterns observed on real-world graphs.

\section{Preliminaries}\label{sec:prelim}
In this section, we provide some mathematical and notational backgrounds.

A \textit{weighted graph} $G = (V, E, W)$ consists of a node set $V = V(G)$, an edge set $E = E(G) \in \binom{V}{2}$, and edge weights $W = W(G)$.
By ignoring the edge weights, we have the underlying \textit{unweighted} graph $\overbar{G} = (V, E)$ of $G$.
For each edge $e \in E$, $W_e$ is the \textit{weight} of $e$.
In this work, we focus on graphs with positive integer edge weights,
i.e., $W \in \natnum^E$, where $\natnum$ is the set of positive integers,
where each edge weight represents the number of occurrences of the corresponding edge.
Note that the analysis on weighted graphs is mainly done on the graphs with integer edge weights~\citep{newman2004analysis},
and for graphs with non-integer edge weights, we may round each edge weight to the nearest integer.
All graphs are assumed to be undirected and without self-loops.
Thus, $(u, v)$ and $(v, u)$ represent the same edge (i.e., the set $\setpar{u, v}$) between two nodes $u \neq v \in V$.

The concepts below use \textit{only the topology} of $G$ (i.e., $V$ and $E$).
For each node $v \in V$, $N_{v}(G) = \setpar{v' \in V: (v, v') \in E}$ is the \textit{neighborhood} of $v$ in $G$, and $d_v(G) = \vert N_{v}(G) \vert$ is the \textit{degree} of $v$ in $G$. 
For two nodes $u, v \in V$, $CN_{uv}(G) = N_{u}(G) \cap N_{v}(G)$ is the set of \textit{common neighbors} (CNs) of $u$ and $v$ in $G$, and we say the two nodes $u$ and $v$ \textit{share} the common neighbors in $CN_{uv}(G)$.
For an edge $e = (u, v)$, we use $CN_e(G)$ to denote $CN_{uv}(G)$, and we say that the edge $e$ \textit{shares} the common neighbors in $CN_e(G)$. Sometimes, the number $\vert CN_e(G) \vert$ of common neighbors shared by the two endpoints of $e$ is called the \textit{embeddedness}~\citep{cleaver2002reinventing} of $e$.

Given $G = (V, E, W)$, the \textit{line graph}~\citep{harary1960some} of $G$ is the graph $L(G) = (E, X)$, where the nodes of $L(G)$ one-to-one correspond to the edges of $G$ and two nodes of $L(G)$ are adjacent to each other if and only if the two corresponding edges in $G$ share a common endpoint.
{Given $k \in \natnum$, the $k$-core~\citep{seidman1983network} $C_k(G)$ of $G$ is the maximal subgraph of $G$ where each node of $C_k(G)$ has degree $k$ within it}.

\begin{table}[t!]
	\begin{center}
		\caption{Notations and abbreviations.}\label{tab:notations}
		{
			\resizebox{\linewidth}{!}{%
				\begin{tabular}{p{0.3\linewidth}  p{0.7\linewidth}}
					\toprule
					\textbf{Notation/Abbreviation} & \textbf{Definition/Meaning}\\
					\midrule
					$G = (V, E, W)$         & a graph with a node set $V$, a edge set $E$, and edge weights $W$ \\
					$\overbar{G} = (V, E)$  & the underlying unweighted graph of $G$ \\
					$W_e$                   & the edge weight of $e \in E$ \\
					$N_v(G)$                & the set of neighbors of $v \in V$ \\
					$d_v(G)$                & the degree of $v \in V$ \\
					$CN_{uv}(G)$            & the set of common neighbors of $u, v \in V$ \\
					$G_i = (V_i, E_i, W_i)$ & the layer-$i$ of $G$ whose edge set consists of the edges with weights $\geq i$ in $G$ \\            
					$R_i = \binom{V_i}{2}$  & the set of all the node pairs in $G_i$ \\				
					$E_{c; i}(G)$           & the set of edges sharing $c$ common neighbors in $G_i$ \\
					$R_{c; i}(G)$           & the set of pairs sharing $c$ common neighbors in $G_i$ \\
					$f_{overall; i}(G)$ ($\tilde{f}_{overall; i}(G)$)     & the overall fraction of weighty edges (adjacent pairs) w.r.t $G$ and $i$ (Def.~\ref{def:SEs_and_FSEs}) \\
					$f_{c; i}(G)$ ($\tilde{f}_{c; i}(G)$) & the fraction of weighty edges (adjacent pairs) within $E_{c; i}(G)$ (Def.~\ref{def:SEs_and_FSEs}) \\
					\midrule                   
					CN                      & common neighbor \\
					\ours                   & \ourslong \\
					FoWE                    & fraction of weighty edges \\                    
					FoAP                    & fraction of adjacent pairs \\
					\bottomrule
				\end{tabular}
			}
		}
	\end{center}
\end{table}

We list the frequently used notations {and abbreviations} in Table~\ref{tab:notations}.
In the notations, the input graph $G$ can be omitted when the context is clear.

In this paper, we consider the problem where given the topology of a graph, we aim to assign realistic edge weights to the topology based on several patterns observed on real-world graphs, 
{where each edge weight is a positive integer representing the number of occurrences of the corresponding edge.}
We formulate the considered problem (informally at this moment) as follows:

\begin{problem}[Informal]\label{pro:assign_weights}
	Given an unweighted graph $\overbar{G} = (V, E)$, we aim to generate edge weights 
	$W: E \rightarrow \natnum$ that satisfy \textit{a group of realistic properties regarding the interplay between topology and edge weights}, {where each edge weight is a positive integer representing the number of occurrences of the corresponding edge.}
\end{problem}

We shall first present the patterns (i.e., the group of realistic properties) that we observe on real-world graphs,
and then we provide a formal problem statement by formulating the patterns as mathematical properties.
Finally, we propose an algorithm that assigns realistic edge weights to a given topology while preserving the formulated properties.

\section{Proposed concepts}\label{sec:concepts}
In this section, we introduce the proposed concepts.
We will use them to describe our observations and design our algorithm.

When given an unweighted graph $\overbar{G} = (V, E)$, the topology divides the pairs $\binom{V}{2}$ of nodes into two categories.
Each pair $(u, v) \in \binom{V}{2}$ of nodes is either \textit{adjacent} ($(u, v) \in E$, weight $\geq 1$) or \textit{distant} ($(u, v) \notin E$, weight $< 1$).
When we have a weighted graph $G = (V, E, W)$, we can similarly set different weight thresholds $i \in \natnum$ and extract the subgraph consisting of edges with weight $\geq i$, which gives the following definition of \textit{layers}.
\begin{definition}[Layers]\label{def:layers}
	Given $G = (V, E, W)$ and $i \in N$, the \textit{layer}-$i$ of $G$ is the weighted graph $G_i = (V_i, E_i, W_i)$ obtained from $G$ by taking the edges with weights {greater than or equal to} $i$.\footnote{$G_1$, the layer-$1$ of $G$, is identical to the original graph $G$.}
	Formally, $E_i = E_i(G) = \setpar{e \in E: W_e \geq i}$, and $W_i = W_i(G)$ satisfies that $W_i(e) = W(e), \forall e \in E_i$.
	{We also define all the possible node pairs $R_i = R_i(G) = \binom{V_i}{2}$.}
\end{definition}
Based on the concept of layers, we also define the following related concepts, weighty edges (WEs) and fraction of weighty edges (FoWE), w.r.t each layer of a graph.
Intuitively, in each layer, the weighty edges are the edges with weights higher than the threshold determined by the layer.
Notably, we define overall FoWEs for the whole layer, and we also define FoWE w.r.t each number $c$ of common neighbors (CNs).
\begin{definition}[Fractions of weighty edges and adjacent pairs]\label{def:SEs_and_FSEs}
	Given $G = (V, E, W)$ and $i \in \natnum$, we call an edge $e \in E_i$ a \textit{weighty edge} (w.r.t $G$ and $i$) if and only if $W_{e} > i$ (i.e., $e \in E_{i+1}$), where recall that $E_i$ is the edge set of $G_i$.
	The \textit{overall fraction of weighty edges} $f_{overall;i}(G)$ is defined as $\vert E_{i + 1}\vert / \vert E_{i}\vert$.
	Further given $c \in \natnum$, let $E_{c;i} \subseteq E_i$ denote the set of edges sharing $c$ CNs (i.e., edges whose endpoints share $c$ CNs) in $G_i$.
	The \textit{fraction of weighty edges} (w.r.t $G$, $i$, and $c$) $f_{c;i}(G)$ is defined as $\vert E_{c;i} \cap E_{i+1}\vert / \vert E_{c;i}\vert$.
	{Similarly, we define the \textit{overall fraction of adjacent pairs} $\tilde{f}_{overall; i}(G) = \vert E_i \vert / \vert R_i \vert$, as well as the \textit{fraction of adjacent pairs} (w.r.t $G$, $i$, and $c$) $\tilde{f}_{c; i}(G) = \vert E_{c; i} \vert / \vert R_{c; i} \vert$ for each $c$, where $R_{c; i}$ is the set of pairs sharing $c$ CNs in $G_i$.}\footnote{The fraction of adjacent pairs can be different from the \textit{density} of the corresponding induced subgraph since $R_{c; i}$'s are defined w.r.t pairs not nodes.}
\end{definition}

\begin{table}[t!]
	\begin{center}
		\caption{Some basic statistics (number of nodes and number of edges in each of the first four layers) of the eleven real-world datasets~\citep{opsahl2013triadic,benson2018simplicial,paranjape2017motifs,sinha2015overview} from five domains used in our empirical study. The datasets are grouped w.r.t their domains.}
		\label{tab:datasets}
		\resizebox{0.8\linewidth}{!}{%
			\begin{tabular}{ l|c|c|c|c|c }
				\toprule
				\textbf{dataset} & $\vert V\vert$ & $\vert E\vert = \vert E_1\vert$ & $\vert E_2\vert$ & $\vert E_3\vert$ & $\vert E_4\vert$ \\
				\midrule
				OF & 897 & 71,380 & 47,266 (66.2$\%$) & 35,456 (49.7$\%$) & 28,546 (40.0$\%$) \\
				\midrule
				FL & 2,905 & 15,645 & 4,608 (29.5$\%$) & 1,507 (9.6$\%$) & 564 (3.6$\%$) \\
				\midrule
				th-UB & 82,075 & 182,648 & 7,297 (4.0$\%$) & 2,090 (1.1$\%$) & 965 (0.5$\%$) \\
				th-MA & 152,702 & 1,088,735 & 128,400 (11.8$\%$) & 48,605 (4.5$\%$) & 26,121 (2.4$\%$) \\
				th-SO & 2,301,070 & 20,989,078 & 1,168,210 (5.6$\%$) & 350,871 (1.7$\%$) & 170,618 (0.8$\%$) \\
				\midrule
				sx-UB & 152,599 & 453,221 & 135,948 (30.0$\%$) & 56,115 (12.4$\%$) & 28,029 (6.2$\%$) \\
				sx-MA & 24,668 & 187,939 & 74,493 (39.6$\%$) & 36,604 (19.5$\%$) & 21,364 (11.4$\%$) \\
				sx-SO & 2,572,345 & 28,177,464 & 9,871,784 (35.0$\%$) & 4,137,454 (14.7$\%$) & 2,055,034 (7.3$\%$)\\
				sx-SU & 189,191 & 712,870 & 216,296 (30.3$\%$) & 82,475 (11.6$\%$) & 37,655 (5.3$\%$)\\
				\midrule
				co-DB & 1,654,109 & 7,713,116 & 2,269,679 (29.4$\%$) & 1,085,489 (14.1$\%$) & 654,182 (8.5$\%$) \\
				co-GE & 898,648 & 4,891,112 & 1,055,077 (21.6$\%$) & 446,833 (9.1$\%$) & 246,944 (5.1$\%$) \\
				\bottomrule %
			\end{tabular}
		}
	\end{center}
\end{table}

\section{Patterns in real-world graphs}\label{sec:patterns}
In this section, we analyze {eleven} real-world graphs from different domains and extract patterns w.r.t the interplay between topology and edge weights.

\smallsection{Datasets.}
We use 11 {publicly-available} real-world datasets from five different domains.
In Table~\ref{tab:datasets}, we give some basic statistics (the number of nodes and the number of edges in each of the first four layers) of the datasets we study in this work.
In all the datasets, the edge weights can be interpreted as the time of occurrences of the corresponding binary relation.
We take the largest connected component of each graph.
In the \textit{OF} dataset~\citep{opsahl2013triadic}, the nodes are users and an edge represents communication within a {blog} post.
In the \textit{FL} (flights) dataset~\citep{opsahl2011anchorage}, the nodes are airports and an edge represent a flight between two airports.
In the \textit{th} (threads) datasets~\citep{benson2018simplicial}, the nodes are users and an edge exists between two users if they participate in the same thread within 24 hours.
The \textit{sx} (stack exchange) datasets~\citep{paranjape2017motifs} are extracted from the same websites as the \textit{th} datasets, but here an edge exists if one user answers or comments on a question of another, and the two groups of datasets are essentially different (see also Table~\ref{tab:datasets} for the statistical difference).
In the \textit{co} (coauthorship) datasets~\citep{benson2018simplicial, sinha2015overview}, the nodes are authors and an edge exists between the two authors if they coauthor a paper.

\subsection{Why the number of common neighbors?}\label{subsec:why_cns}
First, we shall show that the numbers of common neighbors (CNs) are consistently indicative of edge weights even when compared with the more complicated ones.
We compare the numbers of CNs with several other quantities widely used in link prediction \citep{martinez2016survey} and edge-weight prediction \citep{fu2018link}.
Notably, the number of CNs shared by two adjacent nodes is equal to the number of triangles involving the two nodes.
Real-world graphs are rich in triangles~\citep{tsourakakis2008fast, shin2020fast}.
For special graphs, e.g., bipartite graphs where no triangle exists, we can consider \textit{butterflies} ($(2, 2)$-bicliques) instead~\citep{sanei2018butterfly}.

\begin{table*}[t!]
	\begin{center}
		\caption{\textbf{The numbers of common neighbors are simple yet indicative.}
			We report the point-biserial correlation coefficients between the sequences of each quantity and 
			the binary indicators of repetition. The two strongest correlations are highlighted in bold.
			{See Appendix~\ref{app:why_cn_auc} for the results measured by the area under the ROC curve (AUC).}
			Among all the considered quantities, the number of common neighbors is simplest while achieving the highest average rank in average over all the datasets.
		}
		\label{tab:pearson-rep}
		\resizebox{0.9\linewidth}{!}{%
			\begin{tabular}{ l *{14}{|c} }
				\toprule
				\textbf{dataset} & NC & SA & JC & HP & HD & SI & LI & AA & RA & PA & FM & DL & EC 
				& LP \\
				\midrule
				OF & 0.33 & 0.20 & 0.20 & -0.02 & 0.21 & 0.21 & -0.13 & \textbf{0.34} & \textbf{0.35} & 0.33 & 0.11 & 0.32 & 0.26 & 0.33 \\
				FL & 0.32 & 0.26 & 0.26 & 0.19 & 0.24 & 0.26 & -0.06 & \textbf{0.35} & \textbf{0.35} & 0.21 & 0.08 & 0.18 & 0.17 & 0.31 \\
				th-UB & \textbf{0.48} & 0.02 & 0.00 & 0.03 & 0.00 & 0.01 & -0.05 & 0.47 & 0.40 & 0.33 & 0.26 & 0.21 & 0.37 & \textbf{0.48} \\
				th-MA & \textbf{0.45} & 0.22 & 0.15 & 0.09 & 0.15 & 0.18 & -0.05 & 0.44 & 0.35 & 0.33 & 0.40 & 0.25 & 0.38 & \textbf{0.46} \\
				th-SO & \textbf{0.38} & 0.11 & 0.08 & 0.06 & 0.08 & 0.09 & -0.03 & \textbf{0.39} & 0.33 & 0.22 & 0.33 & 0.18 & 0.26 & 0.37 \\
				sx-UB & \textbf{0.15} & 0.11 & 0.08 & 0.09 & 0.07 & 0.08 & -0.00 & 0.13 & 0.10 & 0.09 & 0.12 & 0.09 & 0.14 & \textbf{0.15} \\
				sx-MA & \textbf{0.25} & 0.24 & 0.21 & 0.12 & 0.19 & 0.21 & -0.02 & \textbf{0.25} & 0.22 & 0.19 & 0.19 & 0.16 & 0.20 & 0.25 \\
				sx-SO & 0.10 & \textbf{0.11} & 0.08 & 0.08 & 0.07 & 0.08 & 0.00 & 0.10 & 0.07 & 0.05 & 0.10 & 0.07 & 0.07 & \textbf{0.10} \\
				sx-SU & \textbf{0.14} & 0.11 & 0.08 & 0.09 & 0.07 & 0.08 & -0.00 & 0.12 & 0.08 & 0.08 & 0.11 & 0.08 & 0.13 & \textbf{0.15} \\
				co-DB & \textbf{0.20} & -0.08 & -0.09 & -0.05 & -0.08 & -0.08 & -0.16 & \textbf{0.22} & 0.20 & 0.03 & 0.06 & 0.07 & 0.14 & 0.20 \\
				co-GE & \textbf{0.30} & -0.07 & -0.08 & -0.08 & -0.07 & -0.06 & -0.16 & \textbf{0.32} & 0.26 & 0.16 & 0.19 & 0.19 & 0.22 & 0.30 \\
				\midrule
				avg. & \textbf{0.28} & 0.11 & 0.09 & 0.05 & 0.09 & 0.10 & -0.06 & \textbf{0.28} & 0.24 & 0.18 & 0.18 & 0.16 & 0.21 & 0.28 \\
				avg. rank & \textbf{2.2} & 7.5 & 10.5 & 10.6 & 10.8 & 8.9 & 14.0 & 2.5 & 5.5 & 8.2 & 7.1 & 8.5 & 6.5 & \textbf{2.4} \\
				\bottomrule 
			\end{tabular}
		}
	\end{center}
\end{table*}

Given a graph $G = (V, E, W)$, for each edge $(u, v) \in E$, we consider the following quantities, using only the topology ($V$ and $E$):
\begin{itemize}
	\item \textbf{NC (Number of common neighbors, also called \textit{embeddedness}}~\citep{cleaver2002reinventing}). $NC_{uv} = \vert CN_{uv}\vert$.
	\item \textbf{SA (Salton index)}~\citep{salton1983introduction}. $SA_{uv} = NC_{uv} / \sqrt{d_u \cdot d_v}$.
	\item \textbf{JC (Jaccard index)}~\citep{levandowsky1971distance}. $JC_{uv} = NC_{uv} / \vert N_u \cup N_v\vert$.
	\item \textbf{HP (Hub-promoted)}~\citep{ravasz2002hierarchical}. $HP_{uv} = NC_{uv} / \min(d_u, d_v)$.
	\item \textbf{HD (Hub-depressed)}~\citep{ravasz2002hierarchical}. $HD_{uv} = NC_{uv} / \max(d_u, d_v)$.
	\item \textbf{SI (S{\o}rensen index)}~\citep{sorensen1948method}. $SI_{uv} = 2NC_{uv} / (d_u + d_v)$.
	\item \textbf{LI (Leicht-Holme-Newman index)}~\citep{leicht2006vertex}. $LI_{uv} = NC_{uv} / (d_u \cdot d_v)$.
	\item \textbf{AA (Adamic-Adar index)}~\citep{adamic2003friends}. $AA_{uv} = \sum_{x \in CN_{uv}} 1 / \log d_{x}$.
	\item \textbf{RA (Resource allocation)}~\citep{zhou2009predicting}. $RA_{uv} = \sum_{x \in CN_{uv}} 1 / d_{x}$.
	\item \textbf{PA (Preferential attachment)}~\citep{albert2002statistical}. $PA_{uv} = d_u \cdot d_v$.
	\item \textbf{FM (Friends-measure)}~\citep{fire2011link}. $FM_{uv} = \vert \setpar{(x \in N_u, y \in N_v): x = y \lor (x, y) \in E}\vert  = NC_{uv} + \vert \setpar{(x \in N_u, y \in N_v): (x, y) \in E}\vert$.
	\item \textbf{DL (Degree in the line graph)}. 
$DK_{uv} = d_u + d_v - 2$.
\item \textbf{EC (Edge coreness)}. The maximum $k \in \natnum$ such that the edge $(u,v)$ is in $C_k(G)$, the $k$-core~\citep{seidman1983network} of $G$.
\item \textbf{LP (Local path index)}. $LP_{uv} = (A^2)_{uv} + \epsilon(A^3)_{uv}$, where $A $ is the adjacency matrix of $G$. We use $\epsilon = 10^{-3}$ as in \citep{zhou2009predicting}.
\end{itemize}
For each dataset and each considered quantity, we collect the sequence of the quantities of the edges and that of the binary indicators of repetition (i.e., having weight $> 1$), and compute the Point-biserial correlation coefficient~\citep{tate1954correlation} between them.\footnote{The point-biserial correlation measures the correlation between a continuous variable and a discrete variable, and it is mathematically equivalent to the Pearson correlation.}
In Table~\ref{tab:pearson-rep}, we report the results.
Among all the considered quantities, the number of CNs is the \textit{simplest} one while having the \textit{highest} average point-biserial correlation coefficient and the \textit{highest} average ranking w.r.t the correlation with edge repetition over all the datasets.
{See Appendix~\ref{app:why_cn_auc} for the results measured by the area under the ROC curve (AUC).}

{For the four smallest datasets (OF, FL, sx-MA, and th-UB), we also use the four additional quantities with relatively high computational costs.
They are (1) edge betweenness~\citep{girvan2002community}, (2) personalized pagerank~\citep{jeh2003scaling},
and two ``node-centrality'' measures in the line graph $L(G)$ (each node in $L(G)$ corresponds to an edge in $G$):
(3) eigenvector centrality~\citep{bonacich1987power} and (4) pagerank~\citep{page1999pagerank}.}
For all four datasets, NC consistently has a higher correlation than the four quantities mentioned above.
Moreover, in line graphs, we also consider the closeness centrality~\citep{freeman1977set}, the betweenness centrality~\citep{freeman1977set}, and the clustering coefficients~\citep{watts1998collective}.
However, due to the even larger computational costs of these quantities, it is only possible to compute them on the smallest dataset \textit{OF}, and the Pearson correlation coefficients are
0.12, 0.06, and -0.12 for the closeness centrality, the betweenness centrality, and the clustering coefficients, respectively, which are much lower than that of NC, even though they are much more complicated than NC.

Below, for the clarity and brevity of the presentation, we may visualize or report results on a small number of datasets, while similar results are obtained across all datasets. The full results on all the datasets are available in the supplementary material~\citep{onlineSuppl}.

\begin{figure*}[t!]
\centering
\includegraphics[width=\linewidth]{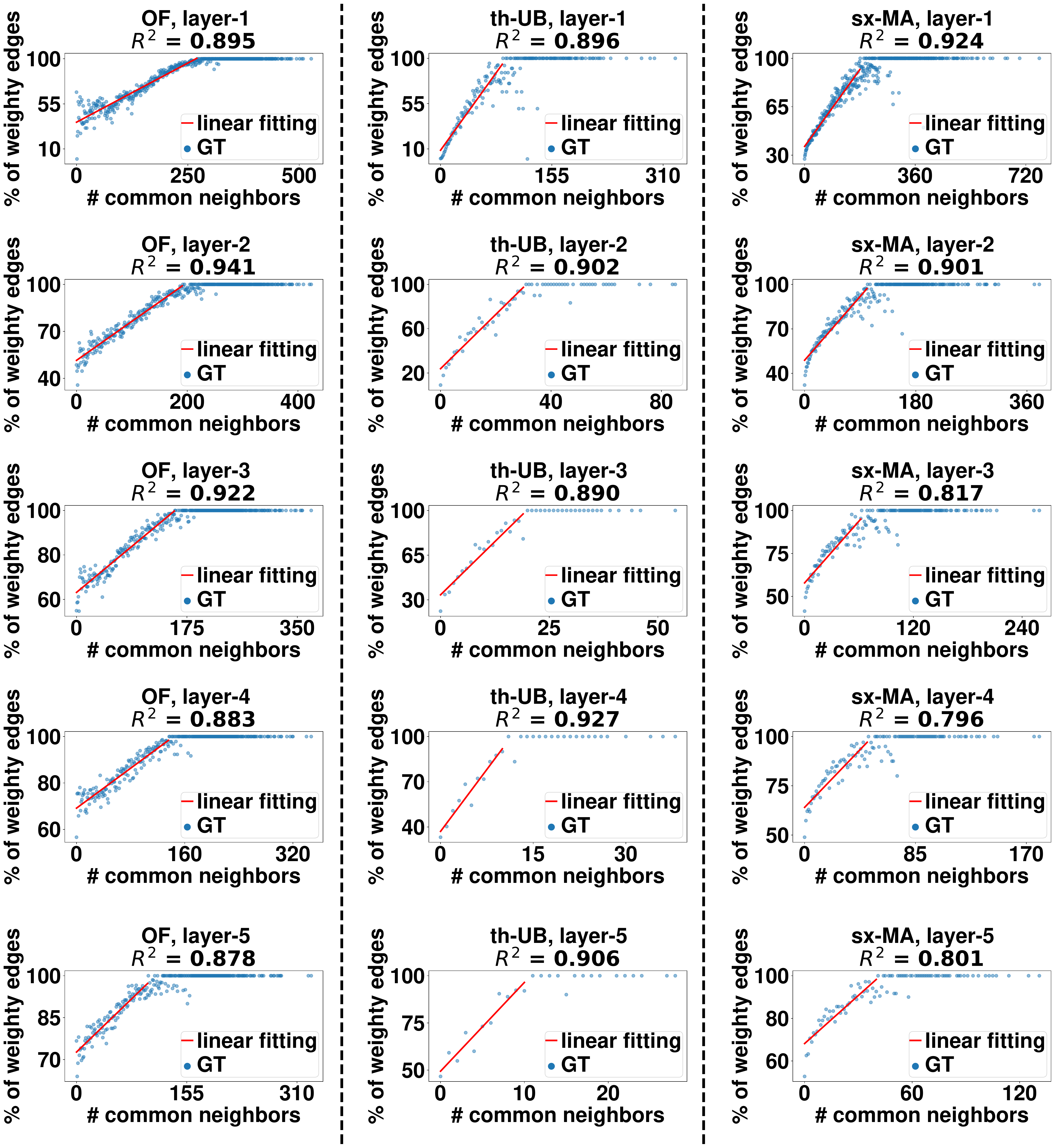}
\caption{\textbf{Fractions of strong edges grow nearly linearly until a saturation point.}
	We also report the $R^2$ of linear fitting before the saturation point.
	The $R^2$ is high in each layer, indicating that the growth is consistently nearly linear.    
	See the supplementary material~\citep{onlineSuppl} for the full results.}
\label{fig:linear-sfs}
\end{figure*}

\subsection{Observation 1: the fractions of weighty edges}\label{subsec:linear-within-layer}
We have shown that the numbers of CNs and the \textit{repetition} (i.e., \textit{weightiness} in layer-$1$) of edges are highly correlated.
We examine this phenomenon in more layers and study the detailed numerical relations between the number $c$ of CNs and the corresponding fraction of weighty edges (FoWE) $f_{c; i}$.
In Figure~\ref{fig:linear-sfs}, for each dataset and each layer-$i$ with $1 \leq i \leq 5$, we plot the FoWEs,
where we can observe that the FoWEs grow \textit{nearly linearly} with the number of CNs until some \textit{saturation point} {(see Definition~\ref{def:saturation_pts_fowes})} such that the FoWEs after the saturation point are almost $100\%$.
In Figure~\ref{fig:linear-sfs}, we also show the results of the linear fitting for the points truncated before the corresponding saturation point with the $R^2$ values, where we can see consistently strong linear correlations.
We formally define the \textit{saturation point} of the fractions of weighty edges (FoWEs)
as the minimum number $c^*$ such that all the edges in $E_{c^*; i}$ are weighty edges.

\begin{definition}[Saturation points of the fraction of weighty edges]\label{def:saturation_pts_fowes}
Given $G$ and $i$, the \textit{saturation point} $c_i^*(G)$ of the fractions of weighty edges is defined as $\min \setpar{c \in \natnum: f_{c; i} = 1}$.
\end{definition}

\begin{remark}
Theoretically, the above definition of saturation point may appear less robust since a single edge that is not weighty can affect the whole group of edges sharing the same number of CNs.
We use such a definition for simplicity and clarity.
In Appendix~\ref{app:saturation_pts}, we discuss this issue and show the practical reasonableness of this definition on the datasets used in our empirical evaluation.
\end{remark}

\color{black}

\begin{observation}[Nearly linear growth of FoWEs]\label{obs:foses_linear_grow}
On each dataset, in each layer, the FoWEs grow nearly linearly with the number of CNs and become almost all $100\%$ after some saturation point.
\end{observation}

\begin{figure*}[t!]
\centering
\includegraphics[width=\linewidth]{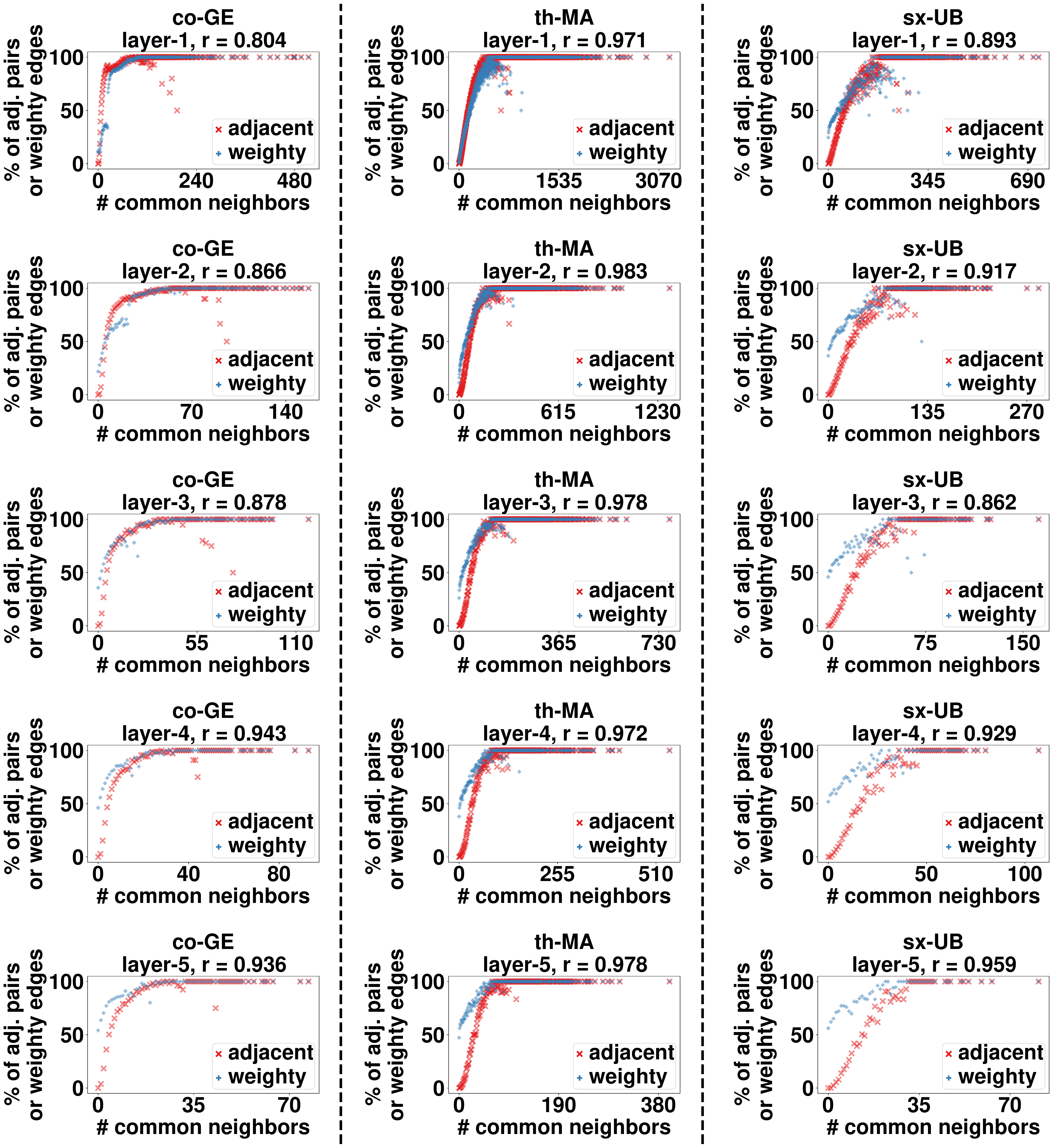}    
\caption{\textbf{The fractions of {(a) adjacent pairs and (b) weighty edges} within each group with the same number of common neighbors are similar.}
	We report Pearson's $r$ between the two fractions, which is high in each layer.    
	{Note that for each dataset, the range of the $x$-axis changes over layers, which is because the maximum number of common neighbors changes over layers. Also, we do not compare the fractions across different layers but only compare them within each layer.}
	The full results are in the supplementary material~\citep{onlineSuppl}.}
\label{fig:con-and-rep}
\end{figure*}

\subsection{Observation 2: adjacency and weightiness}\label{subsec:con-str}
The number of CNs has been widely used for \textit{link prediction} \citep{liu2011link, gunecs2016link}, i.e., inferring the adjacency between node pairs.
In the above Observation~\ref{obs:foses_linear_grow}, we have shown the connection between the number of CNs and the weightiness of edges.
Are the adjacency of pairs and the weightiness of edges also quantitatively related?
In Figure~\ref{fig:con-and-rep}, for each dataset and each layer-$i$ with $1 \leq i \leq 5$, we report how 
(a) the fraction of adjacent pairs within each group of pairs (i.e., $\tilde{f}_{c; i}$) and 
(b) the fraction of weighty edges within each group of edges (i.e., $f_{c; i}$)
depend on the number of CNs, where consistently high Pearson correlation coefficients are observed.
We summarize our observation w.r.t this similarity as follows.

As we have mentioned, the information of CNs has been used for link prediction ~\citep{wang2015link} and for indicating the significance of individual edges~\citep{ahmad2020missing,cao2015grarep,zhu2016link},
where the assumption is usually qualitative, e.g., node pairs between two nodes sharing more CNs are more likely to be adjacent (or more important).
However, no existing works study the \textit{quantitative} relation between the adjacency of pairs and the weightiness (repetition) of edges w.r.t the number of CNs in a unified way and compare them with each other.

\begin{table}[t!]
\begin{center}
	\caption{\textbf{The saturation point of the fractions of weighty edges and that of the fractions of adjacent pairs are consistently close.}
		For each layer of each dataset, we report the saturation point of the fractions of adjacent pairs on the left and that of the fractions of weighty edges on the right.}
	\label{tab:sf-con-str}
	\resizebox{0.6\linewidth}{!}{%
		\begin{tabular}{ lccccc }
			\toprule
			\textbf{dataset} & layer-$1$ & layer-$2$ & layer-$3$ & layer-$4$ & layer-$5$ \\
			\midrule
			OF              & 241/271 & 190/192 & 157/156 & 134/137 & 119/101 \\ 
			\midrule
			FL          & 64/66 & 31/31 & 17/17 & - & - \\
			\midrule
			th-UB       & 73/87 & 30/31 & 19/20 & 18/11 & 15/11 \\
			th-MA       & 372/401 & 145/153 & 114/114 & 84/67 & 63/59 \\
			th-SO       & 685/750 & 208/205 & 134/129 & 97/82 & 74/72 \\
			\midrule
			sx-UB       & 152/149 & 63/69 & 48/42 & 36/27 & 31/22 \\
			sx-MA       & 185/181 & 113/102 & 75/63 & 60/49 & 51/41 \\
			sx-SO       & 886/749 & 407/324 & 221/203 & 169/130 & 120/103 \\
			sx-SU       & 202/206 & 96/93 & 63/54 & 48/37 & 36/27 \\
			\midrule
			co-DB       & 83/88 & 36/29 & 22/24 & 20/21 & 16/16 \\
			co-GE       & 74/92 & 52/49 & 34/40 & 28/30 & 24/21 \\
			\bottomrule %
		\end{tabular}
	}
\end{center}
\end{table}

Similar to the saturation point of FoWEs,
we also defined the saturation point of the fractions of adjacent pairs (FoAPs)
as the minimum number $c^*$ such that all the pairs in $R_{c^*; i}$ are adjacent pairs.
\begin{definition}[Saturation points of the fractions of adjacent pairs]\label{def:saturation_pts}
Given $G$ and $i$, 
recall that the \textit{saturation point} $c_i^*(G)$ of the fractions of weighty edges is defined as $\min \setpar{c \in \natnum: f_{c; i} = 1}$, 
the \textit{saturation point} $\tilde{c}_i^*(G)$ of the fractions of adjacent pairs is defined as $\min \setpar{c \in \natnum: R_{c; i} = E_{c; i}}$.
\end{definition}
As shown in Table~\ref{tab:sf-con-str}, we observe that the saturation point of the FoWEs is consistently similar to that of the FoAPs (see Figure~\ref{fig:con-and-rep}).
\begin{observation}[Similarity between pair-adjacency and edge-weightiness]\label{obs:adj_and_str}
On each dataset, in each layer, the trends of the fractions of adjacent pairs and the fractions of weighty edges w.r.t the number of CNs have a high correlation {(see the consistently high Pearson's $r$ values)},\footnote{{We are studying the correlations here, and the absolute differences are not necessarily small.}} and the saturation point of the FoWEs is close to that of the FoAPs.
\end{observation}

\subsection{Observation 3: a power law across layers}\label{subsec:powerlaw-sf0-overall}
The previous observations describe some patterns within each layer.
Is there any pattern that connects different layers?
For each layer-$i$, we collect the information of $f_{0; i}$ (FoWE of the group of edges without CNs in layer-$i$) and $f_{overall; i}$ (the overall FoWE of all the edges in layer-$i$).
By doing so, we obtain two sequences ($f_{0; i}$'s and $f_{overall; i}$'s) across different layers.
We observe a consistent and strong power law between the two sequences, which we visualize in Figure~\ref{fig:p0-p-powerlaw}.
In the figure, for each dataset, we plot (a) the point $(f_{overall; i}, f_{0; i})$ for each $1 \leq i \leq 10$ in the log-log scale and (b) the power-law fitting line,\footnote{We only include the first four layers of \textit{FL} since the layer-$5$ is too sparse and small.} which is linear in the log-log scale. The consistent and strong power law is clearly observed.
\begin{observation}[A power law across layers]\label{obs:fose_power_law}
On each dataset, across the layers, the FoWEs of the group of edges sharing no CNs and the overall FoWEs of all edges follow a strong power law.
\end{observation}

{
\begin{remark}\label{rem:general_edge_weights}
	All the above observations are based on layer structures. For weighted graphs with positive-integer edge weights, decomposing such graphs into layers is straightforward, while for weighted graphs with real-valued edge weights, we can convert the edge weights into integers by rounding or other ways.
	Moreover, in the datasets used in this work, the edge weights represent the number of occurrences.
	Therefore, the observations may not hold on weighted graphs where the edge weights have other real-world meanings.
	See Appendix~\ref{app:general_edge_weights} for more discussions.
\end{remark}
}
\begin{figure*}[t!]
\centering
\includegraphics[width=\linewidth]{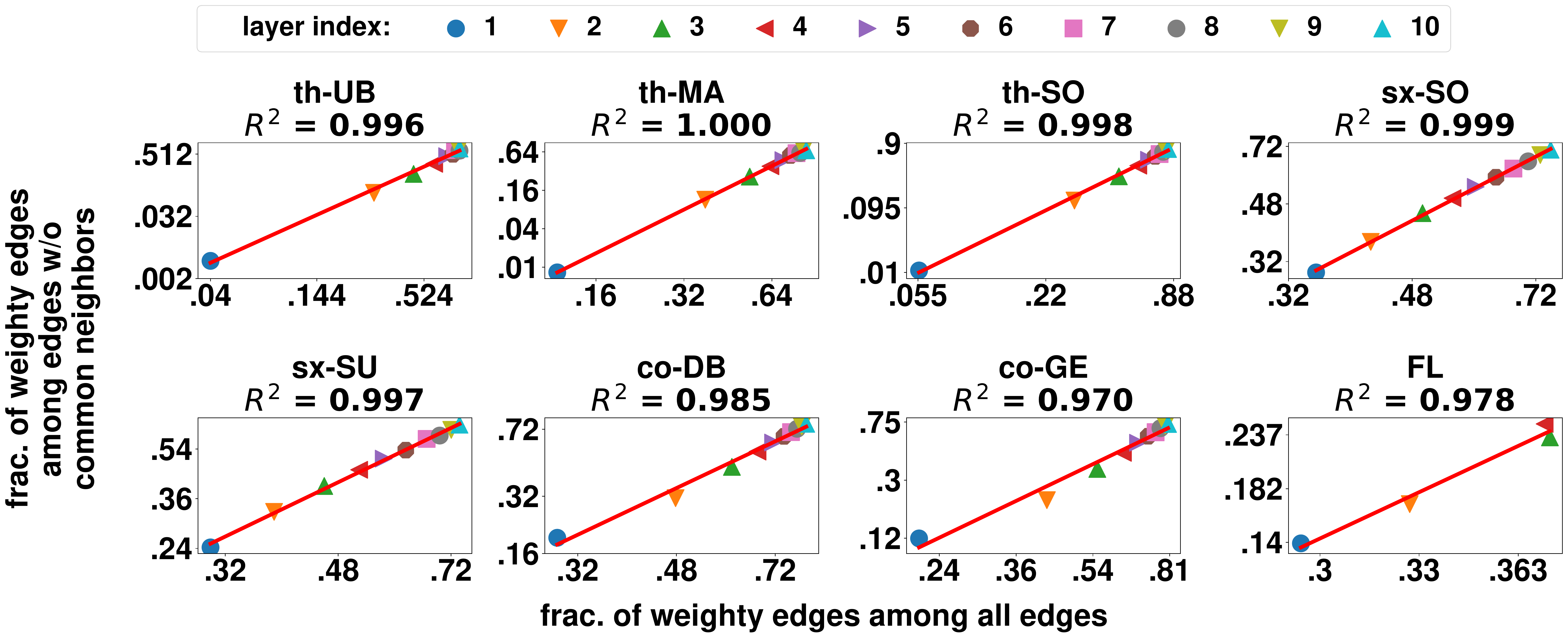}
\caption{\textbf{Consistent and strong power laws exists between the fractions of weighty edges of (a) those without common neighbors and (b) those of all edges.}
	We also report the $R^2$ of power-law fitting.
	The $R^2$ is consistently close to $1$ for all datasets, indicating consistently strong power laws.
	See the supplementary material~\citep{onlineSuppl} for the full results.
}
\label{fig:p0-p-powerlaw}
\end{figure*}

\section{Proposed algorithm: \ours}\label{sec:model}

In this section, we propose \ours (\bus{P}attern-based \bus{E}dge-weight \bus{A}ssignment on g\bus{R}aphs), an algorithm with only \textit{two} parameters that assigns weights to a given topology. 
\ours produces the edge weights layer by layer based on the above observations.
Below, we shall describe the mathematical formulation of our observations and then the detailed procedure of \ours.

\subsection{Formulation of the observations}\label{subsec:math_formulation}
In Section~\ref{sec:patterns}, we describe the patterns that we have observed on the real-world datasets.
For constructing an algorithm, we shall first mathematically formulate the observations.
In the formulation, we may idealize the observations into simple, intuitive, and deterministic formulae.

\smallsection{Linear growth of FoWEs.}
In Observation~\ref{obs:foses_linear_grow}, we have mentioned that the fractions of weighty edges (FoWEs) grow \textit{nearly linearly} with the number of CNs with some saturation points consistently over all datasets and all layers. 
In our algorithm \ours, we assume \textit{perfect linearity}, and we formulate this phenomenon as follows:
\begin{align}\label{eq:foses_linear}
	f_{c; i}(G) = \min(1, f_{0;i}(G) + (1 - f_{0;i}(G)) \frac{c}{c_i^*(G)}), \forall G, i, c,
\end{align}
where $c_i^*(G)$ is the saturation point of FoWEs in $G_i$.
By this formulation, $f_{c; i} = f_{0; i}$ when $c = 0$, and it increases with the number of CNs in a ratio of $(1 - f_{0;i}) / c_i^*$ until $f_{c; i} = 1$ when $c$ reaches $c_i^*$ and stays with $f_{c; i} = 1$ thereafter.

\algnewcommand\algorithmicforeach{\textbf{for each}}
\algdef{S}[FOR]{ForEach}[1]{\algorithmicforeach\ #1\ \algorithmicdo}

\begin{algorithm}[t!]
	\caption{\ours for edge-weight assignment\label{alg:weight-assign}}
	\hspace*{\algorithmicindent} \textbf{Input:} (1) $\overbar{G} = (V, E)$, (2) power-law parameters $a$ and $k$ \\
	\hspace*{\algorithmicindent} \textbf{Output:} weight assignment $W$
	\begin{algorithmic}[1]
		\State $G_1 \leftarrow \overbar{G}$ \label{line:init_1}
		\State $W_e \leftarrow 1, \forall e \in E$ \label{line:init_2}
		\For{$i = 1, 2, 3, \ldots$}
		\If{$\tilde{c}_i^*$ does not exist or $\sum_{c < \tilde{c}_i^*} \vert E_{c; i}\vert = 0$} \label{line:break_condition}
		\State Break \label{line:break}
		\EndIf
		\State $\vert E_{i+1}\vert  \leftarrow$ the solution in $(0, \vert E_i\vert )$ of Equation~\eqref{eq:solve_se} (Property~\eqref{eq:foses_powerlaw}) \label{line:solve_se}
		\State $f_{0; i} \leftarrow a (\vert E_{i+1}\vert  / \vert E_i\vert)^k$ (Eq.~\eqref{eq:p0-se}) \label{line:compute_f0i}
		\State $c_i^* \leftarrow \Tilde{c}_i^*$ (Property~\eqref{eq:sp_close}) \label{line:approx_sp}
		\State $C_i \leftarrow \setpar{\vert CN_e \vert: e \in E_{i}}$ \label{line:set_num_CNs}
		\State $E_c \leftarrow \setpar{e \in E_i: \vert CN_e \vert = c}, \forall c \in C_i$
		\State $f_{c; i} \leftarrow \min(1, f_{0;i} + (1 - f_{0;i}) {c}/{c_i^*}), \forall c \in C_i$ (Property~\eqref{eq:foses_linear}) \label{line:compute_fci}
		\State $E_{i+1} \leftarrow \emptyset$ \label{line:set_se_init}
		\ForEach{$c \in C_i$} \label{line:sample_s}
		\State $\hat{E}_{c} \leftarrow$ sample $f_{c; i} \vert E_c \vert$ (rounded) edges in $E_c$ uniformly at random \label{line:sample}
		\State $E_{i + 1} \leftarrow E_{i + 1} \cup \hat{E}_c$ \label{line:sample_add}
		\EndFor \label{line:sample_e}
		\State $W_e \leftarrow i + 1, \forall e \in E_{i+1}$ \label{line:next_layer_weight_assign}
		\State $G_{i + 1} \leftarrow (\bigcup_{e \in E_{i + 1}}e, E_{i + 1})$ \label{line:next_layer_construct}
		\EndFor
		\State \Return{$W$} \label{line:return}
	\end{algorithmic}
\end{algorithm}

\smallsection{Saturation points.}
In Observations~\ref{obs:adj_and_str} and \ref{obs:foses_linear_grow}, we have mentioned the \textit{similarity} between the trends of the fraction of adjacent pairs and the FoWEs and we have further pointed out that the saturation points of the two kinds of fractions are close.
In the algorithm, we assume \textit{equality} between $\tilde{c}_i^*$ and $c_i^*$, and we formulate it as follows:
\begin{align}\label{eq:sp_close}
	c_i^*(G) = \tilde{c}_i^*(G), \forall G, i,
\end{align}
where recall that 
$\tilde{c}_i^*(G)= \min \setpar{c: R_{c; i}(G) = E_{c; i}(G)}$
and $R_{c; i}(G)$ is the set of node pairs sharing $c$ CNs in $G_i$.

\smallsection{The power law across layers.}
In Observation~\ref{obs:fose_power_law}, we have mentioned that \textit{the FoWEs of the group of edges sharing no CNs and the overall FoWEs of all edges follow a strong power law}.
In the algorithm, assuming a \textit{perfect} power law, we formulate it as follows:
\begin{align}\label{eq:foses_powerlaw}
	f_{0;i}(G) = a(G) f_{overall; i}^{k(G)}, \forall G, i,
\end{align}
where $a(G)$ and $k(G)$ are the two parameters of the power law which may vary for each different graph $G$.

With Properties~\eqref{eq:foses_linear}-\eqref{eq:foses_powerlaw} formulated and explicitly described, we are now ready to re-state Problem~\ref{pro:assign_weights} in a more formal and specific way.
\begin{problem}[Formal]\label{pro:assign_weights_formal}
	Given an unweighted graph $\overbar{G} = (V, E)$, 
	we aim to generate edge weights $W: E \rightarrow \natnum$ that satisfy Properties~\eqref{eq:foses_linear}-\eqref{eq:foses_powerlaw}.
\end{problem}

\subsection{Algorithmic details}
Now we are ready to describe the algorithmic details of \ours (Algorithm~\ref{alg:weight-assign}), which combines all the above formulae {to produce the edge weights layer by layer.}
From now on, we suppose that $G$ is given as an input and thus fixed.
{By Equation~\eqref{eq:foses_powerlaw}} and the definition of $f_{overall; i} = \vert E_{i + 1} / E_i \vert$, we have
\begin{align}\label{eq:p0-se}
	f_{0; i} = a (\vert E_{i+1} \vert / \vert E_i \vert)^k, \forall i.
\end{align}
For each layer-$i$, the total number $\vert E_{i+1}\vert $ of weighty edges should be equal to the summation of the numbers of weighty edges in all the $E_{c; i}$'s.
By Equations~\eqref{eq:foses_linear} and {\eqref{eq:sp_close}}, it gives
\begin{align}\label{eq:se_counting}
	\vert E_{i+1}\vert 
	=& \sum_c \vert E_{c;i}\vert  f_{c;i} 
	= \sum_c \vert E_{c;i}\vert  \min(1, f_{0;i} + (1 - f_{0;i}) {c} / {c_i^*}) \nonumber \\
	=& \sum_c \vert E_{c;i}\vert  \min(1, f_{0;i} + (1 - f_{0;i}) {c} / {\tilde{c}_i^*}).
\end{align}
Note that $\tilde{c}_i^*$ can be obtained from the given topology when $i = 1$ or from the currently generated layers when $i > 1$ (Line~\ref{line:approx_sp}).
We expand Equation~\eqref{eq:se_counting} to get
\begin{align}\label{eq:se_counting_rephrase}
	\vert E_{i + 1}\vert  
	=& \sum_{c \leq \tilde{c}_i^*} \vert E_{c;i}\vert  (f_{0;i} + (1 - f_{0;i}) {c} / {\tilde{c}_i^*}) + \sum_{c > \tilde{c}_i^*} \vert E_{c;i}\vert  \nonumber \\
	=& f_{0;i} \sum_{c \leq \tilde{c}_i^*} \vert E_{c;i}\vert  (\tilde{c}_i^* - c)/{\tilde{c}_i^*} + \sum_c \vert E_{c;i}\vert  \min(1, {c}/{\tilde{c}_i^*}),
\end{align}
which further gives the relation between $f_{0;i}$ and $\tilde{c}_i^*$: 
\begin{align}\label{eq:p0-sp}
	f_{0;i}
	= \frac{\vert E_{i+1} \vert - \sum_c \vert E_{c;i}\vert  \min(1, {c}/{\tilde{c}_i^*})}{\sum_{c \leq \tilde{c}_i^*} \vert E_{c;i}\vert  (\tilde{c}_i^* - c)/{\tilde{c}_i^*}}.
\end{align}
Equations~\eqref{eq:p0-se} and \eqref{eq:p0-sp} give us two different expressions of $f_{0; i}$ and we can use them to obtain $\vert E_{i+1}\vert $ by solving the following equation:
\begin{align}\label{eq:solve_se}
	a (\vert E_{i+1}\vert  / \vert E_i\vert )^k =
	\frac{\vert E_{i+1}\vert   - \sum_c \vert E_{c;i}\vert  \min(1, {c}/{\tilde{c}_i^*})}{\sum_{c < \tilde{c}_i^*} \vert E_{c;i}\vert  (\tilde{c}_i^* - c){\tilde{c}_i^*}}.
\end{align}
\begin{remark}
	In Equation~\eqref{eq:solve_se}, the denominator is zero only when $E_{c; i} = \emptyset$ for each $c \leq \tilde{c}_i^*$, which implies that all the edges are strong edges and the generated edge weights are not meaningful.
\end{remark}
\color{black}
\begin{theorem}\label{thm:uniq_sol}
	Assume that $\tilde{c}_i^* > 0$ exists and $\sum_{c < \tilde{c}_i^*} \vert E_{c;i}\vert  > 0$.
	If $a < 1$ and $k > 1$, then in the range $(0, \vert E_{i}\vert )$, Equation~\eqref{eq:solve_se} has a unique solution $\vert E_{i+1}\vert  \in (0, \vert E_{i}\vert )$.
\end{theorem}
\begin{proof}
	Define 
	\[
	g(x) = a(\frac{x}{\vert E_i\vert})^k - \frac{x  - \sum_c \vert E_{c;i}\vert  \min(1, {c}/{\tilde{c}_i^*})}{\sum_{c \leq \tilde{c}_i^*} \vert E_{c;i}\vert  (\tilde{c}_i^* - c){\tilde{c}_i^*}}
	\] on $x \in [0, \vert E_{i}\vert ]$.
	We have 
	\[
	g(0) = \frac{\sum_c \vert E_{c;i}\vert  \min(1, {c}/{\tilde{c}_i^*})}{\sum_{c \leq \tilde{c}_i^*} \vert E_{c;i}\vert  (\tilde{c}_i^* - c){\tilde{c}_i^*}} > 0
	\]
	and $g(\vert E_i\vert ) = a - 1 < 0$.
	Since $f$ is continuous, {by the intermediate value theorem (see also Bolzano's theorem),} we have at least one solution in the range $(0, \vert E_i\vert )$.
	For the uniqueness, we have 
	\[
	f''(x) = ak(k-1) x^{k-2} / \vert E_i\vert ^k > 0, \forall x \in (0, \vert E_i \vert)
	\]
	(i.e., $f$ is convex on $(0, \vert E_i\vert )$).
	Assume that we have two roots $0 < x_1 < x_2 < \vert E_i\vert $, let $t = (\vert E_i\vert  - x_2) / (\vert E_i\vert  - x_1) \in (0, 1)$, then by the convexity of $f$ we have 
	\[
	0 = g(x_2) = g(tx_1+(1-t)\vert E_i\vert ) \leq tg(x_1)+(1-t)g(\vert E_i\vert )=(1-t)g(\vert E_i\vert )<0,
	\]
	completing the proof by contradiction.
\end{proof}
After obtaining $\vert E_{i+1}\vert $ (Line~\ref{line:solve_se}), we use Equation~\eqref{eq:p0-se} to compute $f_{0; i}$ (Line~\ref{line:compute_f0i}),
and then use Equation~\eqref{eq:foses_linear} with $c_i^* = \tilde{c}_i^*$ 
to compute all $f_{c;i}$'s {(Lines~\ref{line:approx_sp} and \ref{line:compute_fci}).}
Then for each $c$, we sample $f_{c; i} \vert E_c \vert$ edges uniformly at random in $E_c$ in the current layer $G_i$ to be weighty edges,
assign the edge weights accordingly, and construct the next layer $G_{i+1}$ (Lines~\ref{line:set_se_init}-\ref{line:next_layer_construct}).
We repeat the process layer after layer.
By Theorem~\ref{thm:uniq_sol}, if we have a valid saturation point $\tilde{c}_i^*$ and there exist edges sharing less than $\tilde{c}_i^*$ CNs, we can always obtain a unique solution of $\vert E_{i+1}\vert $ from Equation~\eqref{eq:solve_se}, and thus the whole process can continue.
{If any of the conditions are not met, the process terminates, and the current edge weights are returned as the final output (Lines~\ref{line:break_condition} and {\ref{line:return})}.}

The following theorem shows the time complexity of Algorithm~\ref{alg:weight-assign}.
\begin{theorem}\label{thm:pear_time}
	Given an input graph $\overbar{G} = (V,E)$ and two parameters $a$ and $k$, Algorithm~\ref{alg:weight-assign} takes $O(i_{max} \sum_{v \in V} d_v^2))$ time to output a weight assignment $W$, where $i_{max}$ is the maximum layer index such that $G_{i_{max}}$ is non-empty, i.e., the maximum weight in the output.
\end{theorem}
\begin{proof}
	We shall show that it takes $O(\sum_{v \in V} d_v^2)$ to generate each layer.
	The time complexity consists of
	(1) that of computing $\tilde{c}_i^*$ (Line~\ref{line:approx_sp}) and
	(2) that of sampling weighty edges (Lines~\ref{line:sample_s}-\ref{line:sample_e}).
	For (1), checking the CNs of all pairs can be done by enumerating all neighbor pairs of each node, which takes $O(\sum_{v \in V} d_v^2(G_i))=O(\sum_{v \in V} d_v^2(G))$ time.
	For (2), sampling among $E_i$ takes $O(\vert E_i\vert ) = O(\vert E\vert ) = O(\sum_{v \in V} d_v(G))$ time, completing the proof.
\end{proof}

The following theorem states that Algorithm~\ref{alg:weight-assign} indeed preserves all the formulated properties and is What are the relations between the edge weights and the topology in real-world graphs?
Given only the topology of a graph, how can we assign realistic weights to its edges based on the relations? 
Several trials have been done for \textit{edge-weight prediction} where some unknown edge weights are predicted with most edge weights known.
There are also existing works on generating both topology and edge weights of weighted graphs.
Differently, we are interested in generating edge weights {that are realistic in a macroscopic scope}, merely from the topology, which is unexplored and challenging.
To this end, we explore and exploit the patterns involving edge weights and topology in real-world graphs.
Specifically, we divide each graph into \textit{layers} where each layer consists of the edges with weights at least
a threshold.
We observe consistent and surprising patterns appearing in multiple layers:
the similarity between being adjacent and having high weights,
and the nearly-linear growth of the fraction of edges having high weights with the number of common neighbors.
We also observe a power-law pattern that connects the layers.
Based on the observations, we propose \ours, an algorithm assigning realistic edge weights to a given topology.
The algorithm relies on only \textit{two} parameters, preserves \textit{all} the observed patterns, and produces more realistic weights than the baseline methods with the same number of, or even more, parameters. a valid approach for Problem~\ref{pro:assign_weights_formal}.
\begin{theorem}\label{thm:pear_satisfy}
	Given any $\bar{G} = (V, E)$, $a$, and $k$, the output of \ours (Algorithm~\ref{alg:weight-assign}) satisfies Properties~\eqref{eq:foses_linear}-\eqref{eq:foses_powerlaw} up to integer rounding.
\end{theorem}
\begin{proof}
	Property~\eqref{eq:foses_linear} is explicitly preserved by Line~\ref{line:compute_fci},
	and Property~\eqref{eq:sp_close} is explicitly preserved by Line~\ref{line:approx_sp}.
	Regarding Property~\eqref{eq:foses_powerlaw}, since Property~\eqref{eq:foses_linear} is preserved, the solution of Equation~\eqref{eq:solve_se} which combines Properties~\eqref{eq:foses_linear} and \eqref{eq:foses_powerlaw} preserves Property~\eqref{eq:foses_powerlaw}, completing the proof.
\end{proof}

\begin{remark}
	Due to the interconnectedness of all the three properties, for any fixed $a$ and $k$, the output of \ours is unique up to the sampling (Line~\ref{line:sample_s}-\ref{line:sample_e}), which is evidentially supported by the small standard variations in Tables~\ref{tab:main_res} and \ref{tab:main_res:two}.
\end{remark}

\section{Experiments}\label{sec:eval}
In this section, through experiments on the real-world graphs, we shall show that, in most cases, \ours generates realistic edge weights for a given topology with only \textit{two} parameters and \textit{without sophisticated searching or fine-tuning} on the parameters.

\subsection{Baseline methods and experimental settings}\label{sec:eval:settings}

{The following baseline methods (PRD, SCN, SEB, PEB, and STC) are \textit{unsupervised} in that they do not use any explicit ground truth edge-weight information.
	However, for these methods to output meaningful predictions, we provide the ground-truth number of edges for each layer-$i$ (i.e., $\vert E_i \vert$) with $2 \leq i \leq 5$ to them.
	Such \textit{additional} information is \textit{not} provided to \ours.\footnote{{The $\vert E_i \vert$'s (specifically, $\vert E_1 \vert$, $\vert E_2 \vert$, $\vert E_3 \vert$, and $\vert E_4 \vert$) are essentially \textit{four} parameters, compared to only \textit{two} parameters used in \ours.}}}
\begin{itemize}
	\item \textbf{PRD (purely random).} The PRD method repeatedly uniformly at random chooses an edge and increments its weight until the $\vert E_i \vert$'s are satisfied.
	\item \textbf{SCN (sorting-CN).} Instead of random sampling, the SCN method sorts the edges by the number of CNs and assigns the weights accordingly (higher weights to the edges with more CNs).\footnote{{The CNs are counted in each original graph (i.e., layer-$i$) instead of in each layer.}
		An optimization problem in~\citep{adriaens2020relaxing} of maximizing the total edge weights of all triangles is equivalent to this method.}
	\item \textbf{SEB (sorting-embedding).} The SEB method sorts the edges by the inner product of the node embeddings of the two endpoints of each edge and assigns higher weights to the edges with higher inner products.
	{The node embeddings are produced by two different methods, RandNE~\citep{zhang2018billion, karateclub} and
		node2vec~\citep{grover2016node2vec}.
		We use SEB-R and SEB-N to denote the results using RandNE and node2vec, respectively.}
	The feature dimension is set as $32$, and all the other parameters are kept the same as in the original paper.
	\item \textbf{PEB (probability-embedding).} The PEB method uses node embeddings as in SEB, and it repeatedly chooses an edge and increments its weight until the $\vert E_i\vert $'s are satisfied, where the exponential of the inner product of the node embeddings of the two endpoints of each edge is used as the weight of the edge in the sampling.
	{We use PEB-R and PEB-N to denote the results using RandNE and node2vec, respectively.}
	\item \textbf{STC (strong triadic closure).} The STC method makes use of the strong triadic closure~\citep{sintos2014using} (STC) principle. Specifically, for each layer-$i$, the STC method first uses a greedy algorithm to maximize the number of candidate weighty edges in the layer without having any open triangle (i.e., three nodes $v_1,v_2,v_3$ s.t. the two edges $(v_1, v_2)$ and $(v_1, v_3)$ exist and $(v_2, v_3)$ does not exist).
	After that, the STC method uniformly at random samples $\vert E_i \vert$ weighty edges among all the candidates.\footnote{We simply take all the candidates, if $\vert E_i \vert$ is larger than the number of candidates. We have also tried including all the candidates as the weighty edges, which, however, for each dataset, produced layers that only change slightly after layer-$2$, and thus cannot produce meaningful edge weights more than binary categorization.}
\end{itemize}

By contrast, for the proposed method \ours, we only consider \textit{two} settings $(a, k) \in \setpar{(0.98, 1.02), (0.7, 1.3)}$,
where for the two \textit{co} datasets and the four \textit{sx} datasets we use $(a, k) = (0.98, 1.02)$; and for the remaining datasets we use $(a, k) = (0.7, 1.3)$.
For each dataset, we report the results {in the better setting}.
Note that these settings are chosen without relying on ground-truth weights.
Specifically, when we choose the parameters, we simply move two parameters $a$ and $k$ in the opposite directions, while keeping $a + k = 2$, so that all the candidate parameter settings satisfy the assumptions in Theorem~\ref{thm:uniq_sol}.
Also, the best-performing settings show clear domain-based patterns.
Specifically, we can use the same parameter setting for datasets in the same domain.
{Although it is challenging to find the best-performing setting for a dataset merely based on its topology,\footnote{As a weighted graph evolves, its topology may stay the same, but the edge weights representing the repetitions of edges may change. Such scenarios imply that for a given topology, multiple optimal groups of edge weights exist, and thus it is hard to find the best-performing setting.}
	the structural similarity between datasets within the same real-world domain~\citep{chakrabarti2006graph, wills2020metrics} can be utilized to find a proper setting for each dataset based on the topology.
	See Appendix~\ref{app:choose_params} for more discussions.}

We {also consider} two \textit{supervised} edge-weight prediction methods {directly supervised by ground-truth edge weights.}
{Notably, the supervised methods deal with the prediction task as a classification task and do not rely on the layer structure.}
We give the {below} supervised methods $10\%$ of the ground truth edge weights as the input training set (and another $10\%$ as the validation set if needed).
We make sure each of the training set covers all five classes: edges of weight $1,2,3,4,\geq 5$, corresponding to the five layers we are studying.
Notably, the supervised methods use much more parameters.
The considered methods are:
\begin{itemize}
	\item \textbf{RFF (random forest-feature).} The RFF method uses a random forest classifier~\citep{breiman2001random} with the 14 metrics we have used in Section~\ref{subsec:why_cns} (see Table~\ref{tab:pearson-rep}),
	which follows the procedure in a previous work~\citep{fu2018link} except {for that} the random forest is smaller and some metrics are not used due to its high computational cost as mentioned in Section~\ref{subsec:why_cns}.
	The hyperparameters of random forest are listed as follows:
	the number of trees = 32,
	the maximum depth of the tree = 5,
	and all the other parameters are kept the same as in the original work~\citep{breiman2001random}.
	\item \textbf{NEB (neural network-embedding).} The NEB method uses a neural network consisting of one bilinear layer. For each edge, the node embeddings of both endpoints, {which are obtained as in SEB}, are used as the input, and the neural network is trained to {minimize a classification loss (cross-entropy).}    
	{We use NEB-R and NEB-N to denote the results using RandNE and node2vec, respectively.}
\end{itemize}

\begin{figure*}[t!]
	\centering	\hspace{2mm}\includegraphics[width=0.8\textwidth]{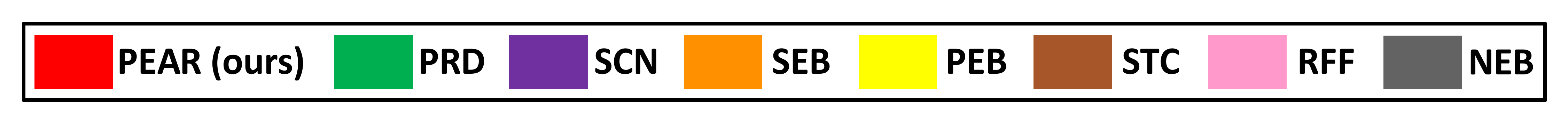}\\
	\includegraphics[scale=0.26]{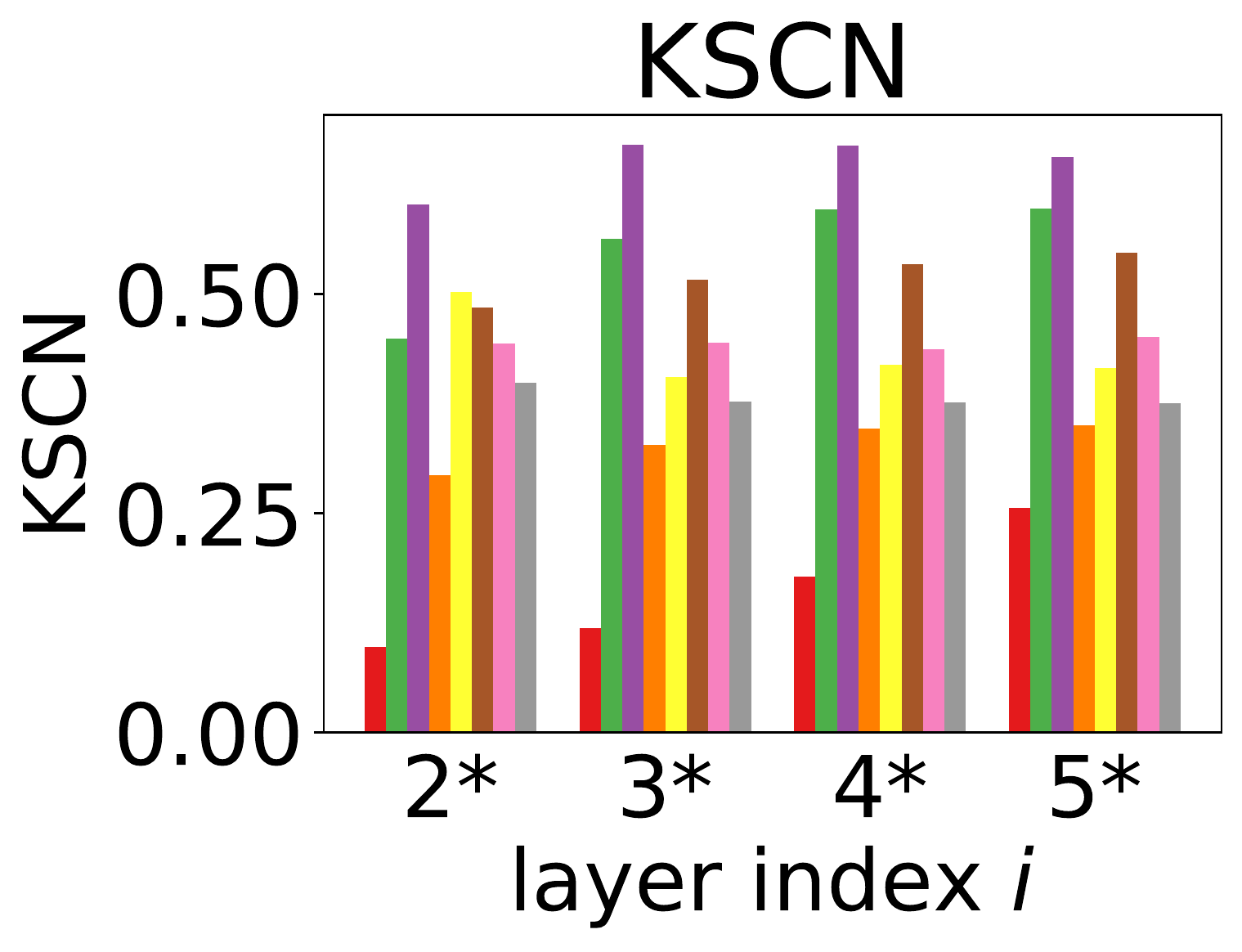}
	\includegraphics[scale=0.26]{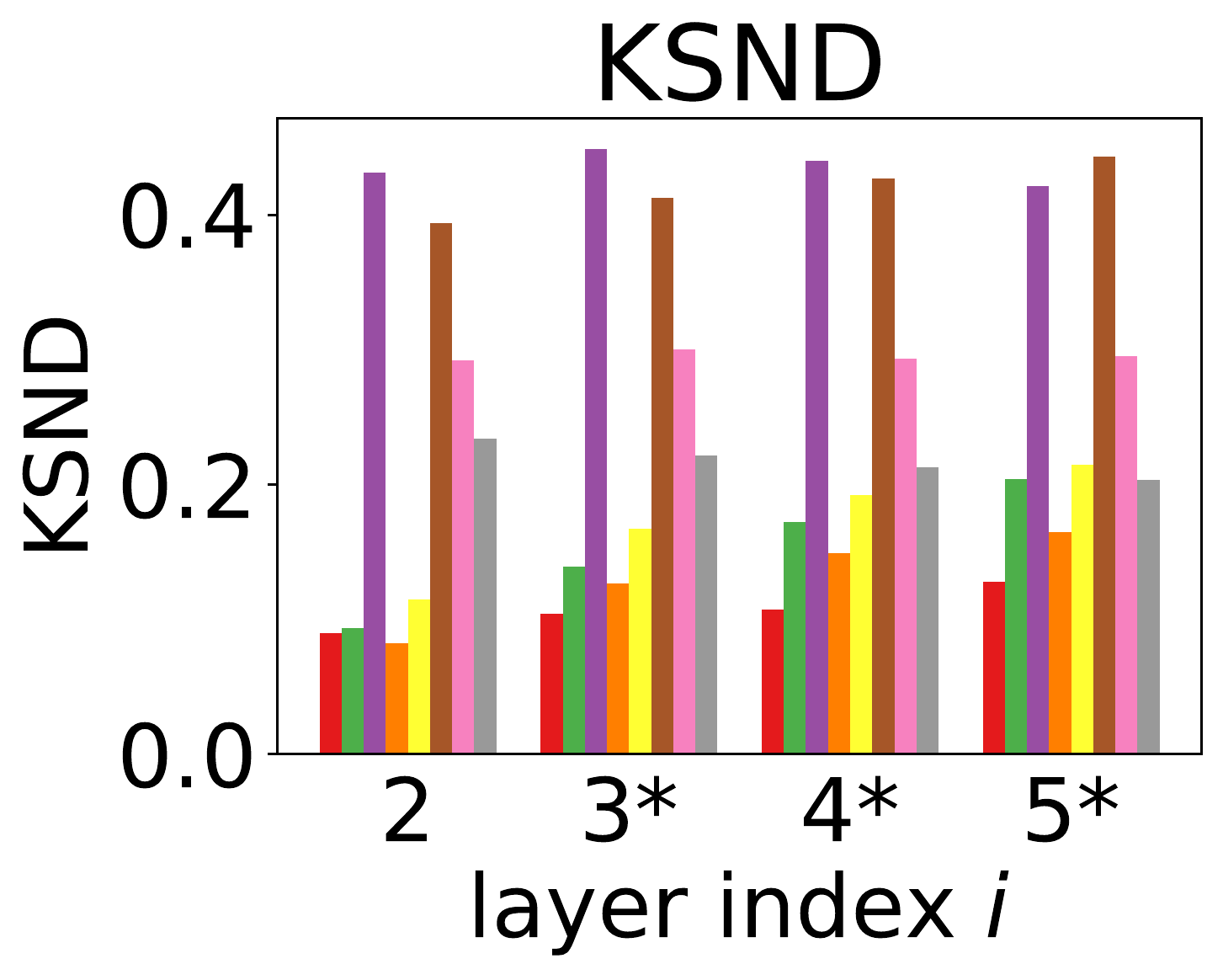}
	\includegraphics[scale=0.26]{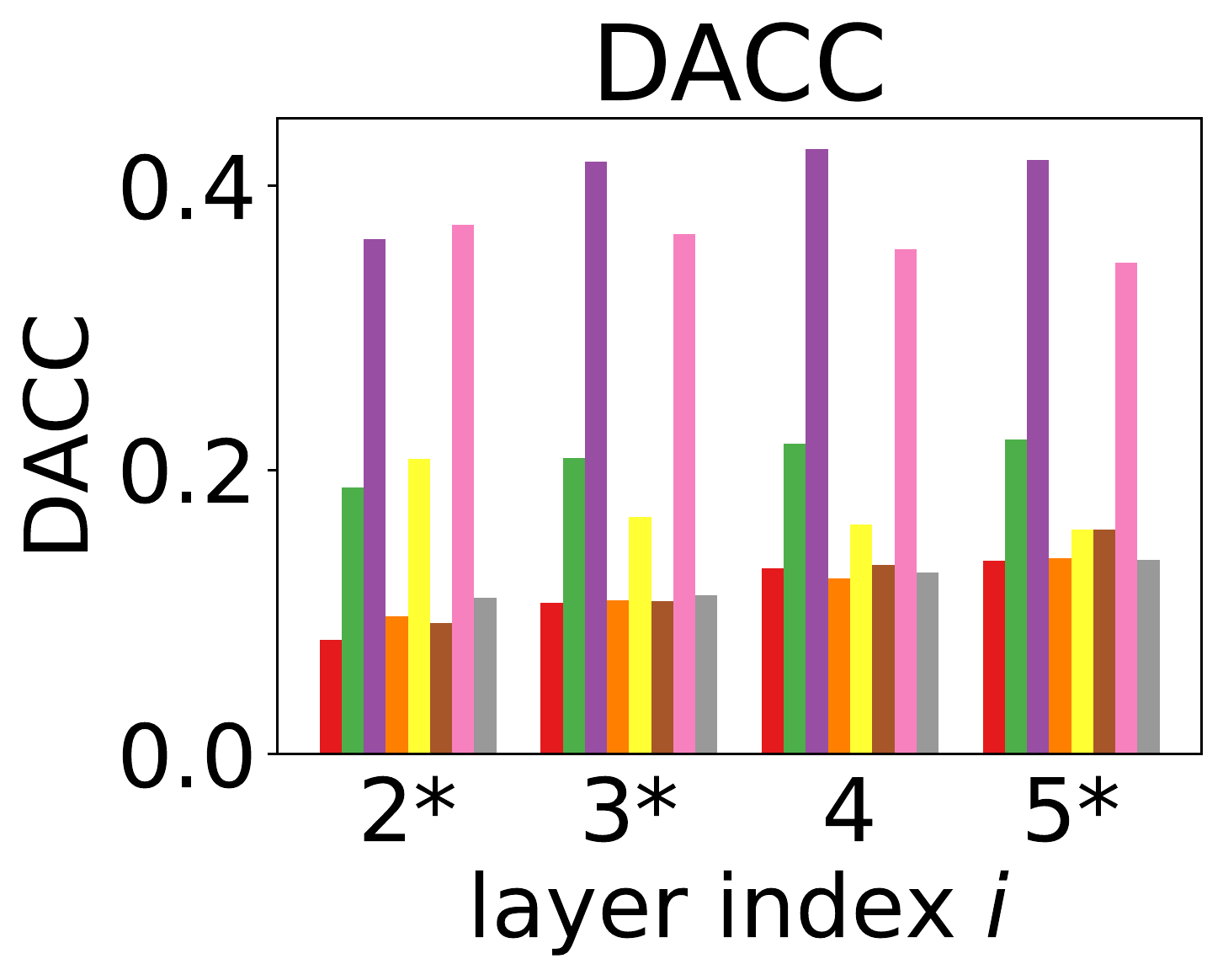}
	\includegraphics[scale=0.26]{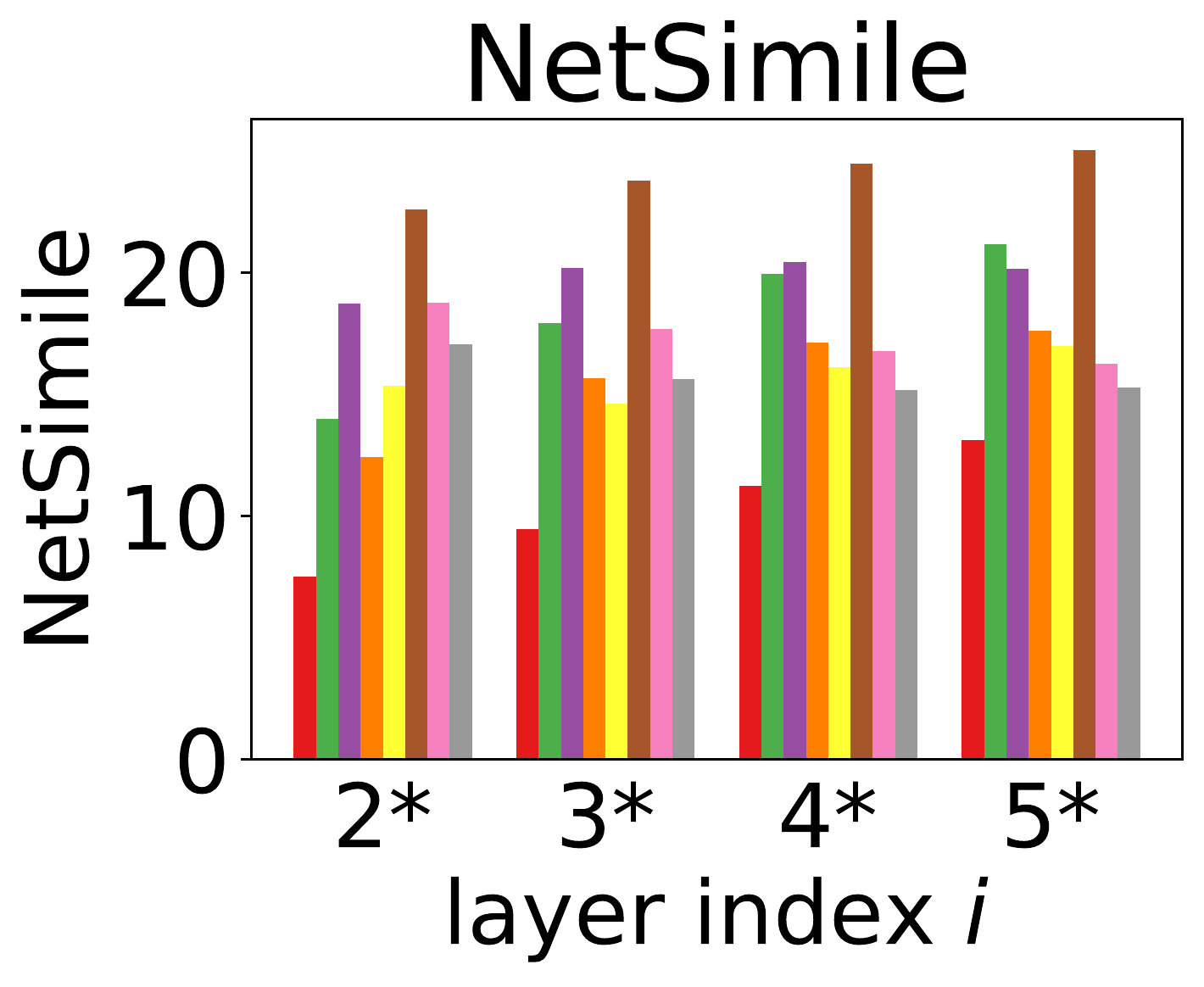}
	\caption{
		\textbf{\ours generates edge weights that are realistic {in} each layer.}
		{For each generated layer-$i$ with $i \in \setpar{2, 3, 4, 5}$, each method, and each metric, we report the average over all datasets.}
		We use asterisk marks ($*$) to indicate the {14 (out of 16)} cases where \ours performs best.
		See Tables~\ref{tab:main_res} and \ref{tab:main_res:two} for the numerical comparison in detail.
	}
	\label{fig:res_by_layer}
\end{figure*}

\begin{figure*}[htb]
	\centering
	\includegraphics[width=0.8\textwidth]{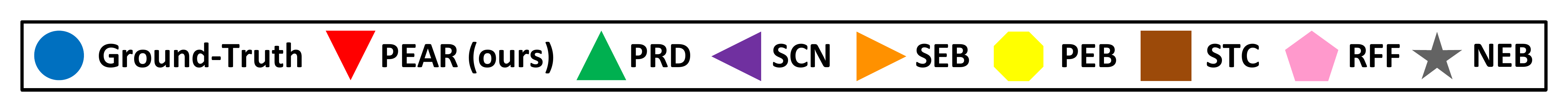}\\
	\begin{subfigure}[b]{0.6\textwidth}
		\centering
		\includegraphics[width=\textwidth]{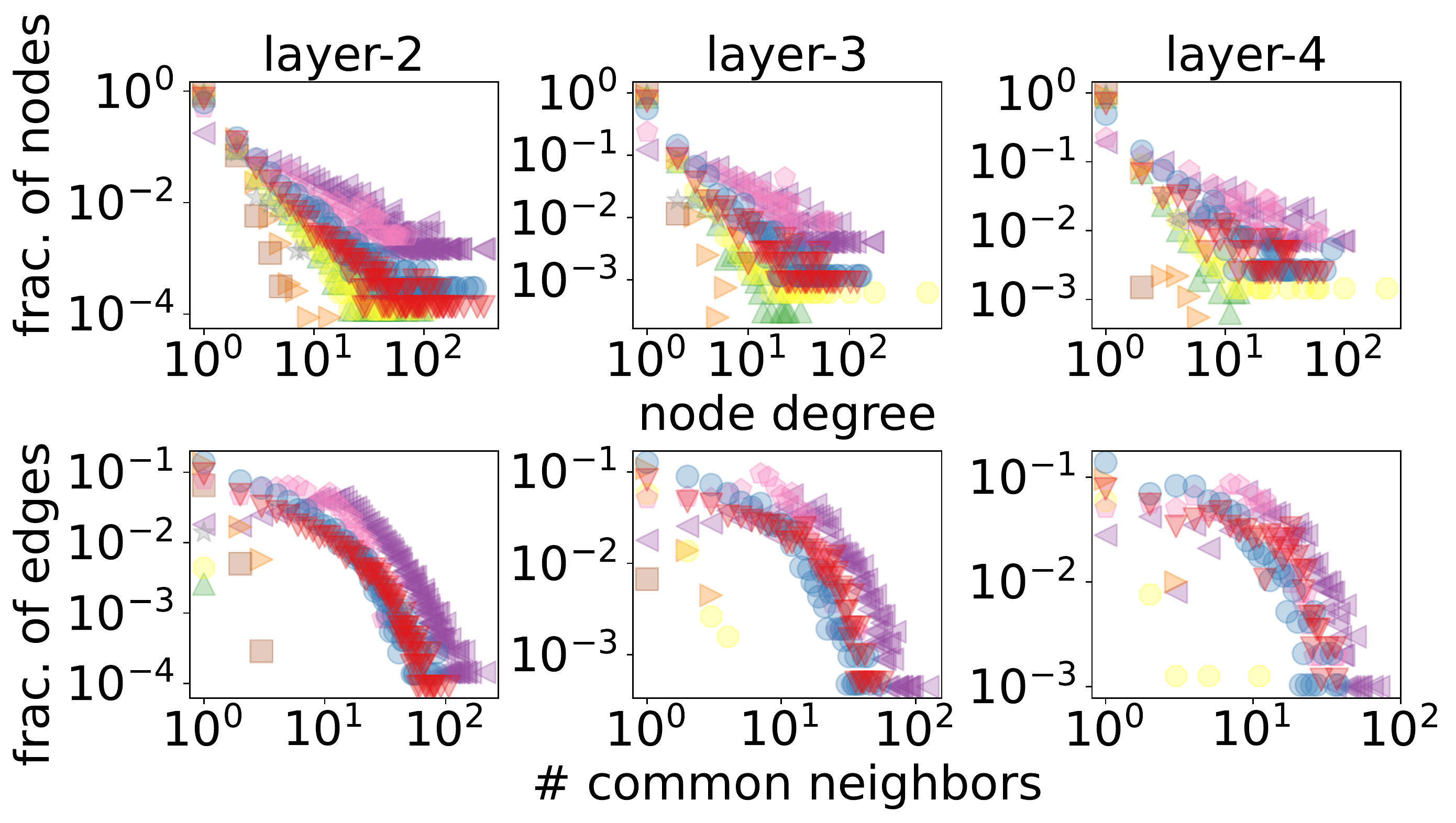}
		\vspace{-6mm}
		\caption{\textit{threads-ubuntu}, $(a, k) = (0.7, 1.3)$}
	\end{subfigure}
	\begin{subfigure}[b]{0.6\textwidth}
		\centering
		\includegraphics[width=\textwidth]{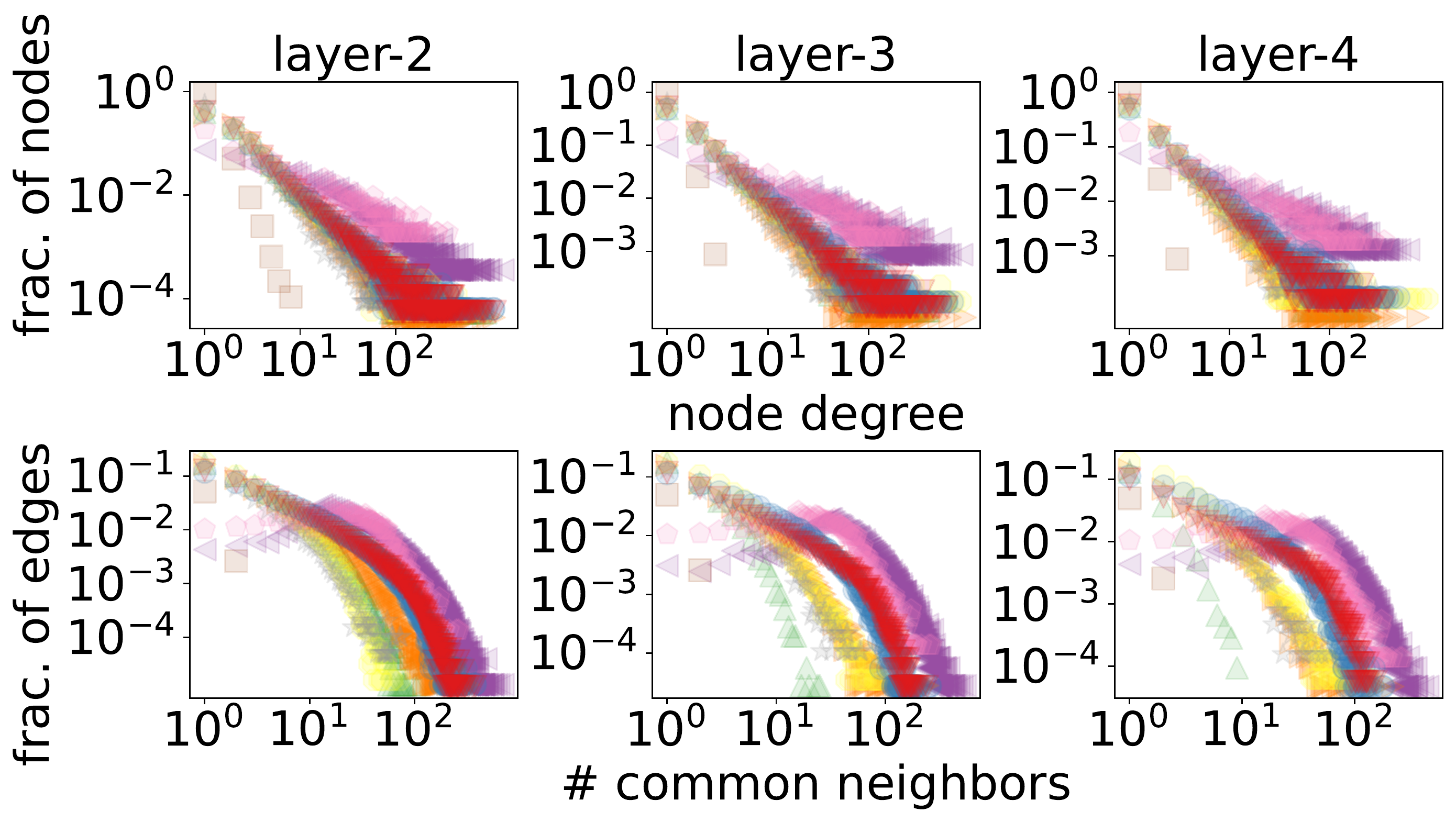}
		\vspace{-6mm}
		\caption{\textit{sx-math}, $(a, k) = (0.98, 1.02)$}
	\end{subfigure}
	\caption{
		\textbf{\ours produces realistic edge weight.}
		For each $i \in \setpar{2, 3, 4}$ and each method including the ground truth, we report the node-degree and number-of-CN distributions of the generated layer-$i$.
		See Tables~\ref{tab:main_res} and \ref{tab:main_res:two} for the numerical comparison in detail.
	}
	\label{fig:evaluation-dist}
\end{figure*}

\begingroup
\setlength{\tabcolsep}{2pt}
\begin{table}
\centering
\caption{\textbf{\ours 
generates more realistic edge weights than the baseline methods in terms of KSCN and KSND.} 
For each dataset and each metric, we report the average over the layers generated by each method, as well as
the average (AVG) over all datasets and the average rank (A.R.), where 
the best result is in \textbf{bold}, and the second-best one is \underline{underlined}
(OOM = out of memory; FAIL = unable to output a meaningful layer).
See the online appendix~\citep{onlineSuppl} for the full results.
}
\label{tab:main_res}                
\begin{adjustbox}{scale=0.75}
\begin{tabular}{ c | c | *{11}{c} | c c }                    
\cmidrule{1-15}
metric & method & OF & FL & th-UB & th-MA & th-SO & sx-UB & sx-MA & sx-SO & sx-SU & co-DB & co-GE & AVG & A.R. \\
\cmidrule{1-15}
\multirow{11}{*}{KSCN}
& PRD & 0.542 & 0.522 & 0.722 & 0.840 & 0.768 & 0.242 & 0.486 & 0.259 & 0.299 & 0.656 & 0.731 & 0.55 & 8.36 \\
& SCN & 0.485 & 0.661 & 0.528 & 0.579 & 0.634 & 0.782 & 0.681 & 0.795 & 0.768 & 0.661 & 0.573 & 0.65 & 8.18 \\
& SEB-R & 0.217 & 0.401 & 0.712 & 0.773 & 0.737 & 0.252 & 0.362 & 0.305 & 0.328 & 0.484 & 0.488 & 0.46 & 7.00 \\
& SEB-N & 0.475 & \bolden{0.184} & 0.670 & 0.601 & 0.624 & \underline{0.100} & \underline{0.221} & 0.159 & 0.237 & \underline{0.155} & 0.199 & \underline{0.33} & \underline{3.73} \\
& PEB-R & 0.626 & 0.368 & 0.665 & 0.711 & 0.607 & 0.109 & 0.310 & \bolden{0.137} & \underline{0.154} & 0.529 & 0.577 & 0.44 & 4.82 \\ 
& PEB-N & 0.626 & 0.368 & 0.665 & 0.711 & 0.607 & 0.109 & 0.310 & \bolden{0.137} & \underline{0.154} & 0.529 & 0.577 & 0.44 & 4.82 \\
& STC & 0.679 & 0.474 & 0.705 & 0.772 & 0.698 & 0.292 & 0.707 & 0.380 & 0.362 & 0.238 & 0.419 & 0.52 & 7.82 \\
& RFF & 0.265 & \underline{0.236} & \underline{0.346} & \underline{0.368} & \underline{0.434} & 0.805 & 0.560 & 0.734 & 0.713 & 0.293 & \underline{0.131} & 0.44 & 5.36 \\
& NEB-R & 0.147 & 0.339 & 0.719 & 0.899 & 0.559 & 0.274 & 0.416 & 0.304 & 0.364 & \bolden{0.096} & \bolden{0.085} & 0.38 & 5.45 \\
& NEB-N & \underline{0.140} & FAIL  & FAIL  & 0.901 & 0.732 & 0.304 & 0.749 & \bolden{0.137} & 0.277 & 0.224 & 0.383 & FAIL & 7.09 \\
\cmidrule{2-15}
& \multirow{2}{*}{\bolden{\ours (ours)}} & \bolden{0.074} & 0.241 & \bolden{0.168} & \bolden{0.186} & \bolden{0.126} & \bolden{0.076} & \bolden{0.072} & 0.239 & \bolden{0.103} & {0.329} & {0.181} & \multirow{2}{*}{\bolden{0.16}} & \multirow{2}{*}{\bolden{2.18}} \\
&  & \bolden{($\pm$0.003)} & ($\pm$0.017) & \bolden{($\pm$0.015)} & \bolden{($\pm$0.006)} & \bolden{($\pm$0.008)} & \bolden{($\pm$0.003)} & \bolden{($\pm$0.001)} & ($\pm$0.001) & \bolden{($\pm$0.002)} & {($\pm$0.013)} & ($\pm$0.008) &  &  \\
\cmidrule{1-15}
\multirow{11}{*}{KSND}
& PRD & 0.096 & 0.164 & 0.336 & 0.219 & 0.161 & 0.040 & \bolden{0.059} & \bolden{0.027} & {0.045} & 0.237 & 0.292 & 0.15 & 4.36\\
& SCN & 0.359 & 0.424 & 0.406 & 0.587 & 0.443 & 0.544 & 0.556 & 0.371 & 0.506 & 0.311 & 0.283 & 0.44 & 9.91 \\
& SEB-R & 0.087 & 0.180 & 0.361 & \underline{0.201} & 0.157 & \underline{0.038} & 0.067 & 0.038 & 0.057 & 0.231 & 0.289 & 0.16 & 4.64 \\
& SEB-N & 0.096 & 0.121 & 0.336 & 0.212 & 0.155 & \bolden{0.027} & \underline{0.065} & \underline{0.034} & \underline{0.041} & 0.155 & 0.198 & \underline{0.13} & \underline{3.09} \\
& PEB-R & 0.232 & 0.172 & 0.289 & 0.248 & \underline{0.133} & 0.041 & 0.085 & 0.071 & 0.057 & 0.274 & 0.293 & 0.17 & 5.27\\
& PEB-N & 0.232 & 0.172 & 0.289 & 0.248 & \underline{0.133} & 0.042 & 0.085 & 0.071 & 0.057 & 0.274 & 0.294 & 0.17 & 5.45 \\
& STC & 0.577 & 0.385 & 0.452 & 0.441 & 0.397 & 0.303 & 0.492 & 0.531 & 0.354 & 0.306 & 0.374 & 0.42 & 10.00 \\
& RFF & 0.115 & 0.146 & \underline{0.267} & 0.388 & 0.322 & 0.602 & 0.416 & 0.246 & 0.410 & {0.217} & {0.121} & 0.30 & 6.64\\
& NEB-R & 0.148 & \underline{0.108} & 0.389 & 0.430 & 0.326 & 0.128 & 0.163 & 0.346 & 0.194 & \bolden{0.055} & \underline{0.113} & 0.22 & 6.18 \\
& NEB-N & \underline{0.085} & FAIL  & FAIL  & 0.327 & 0.352 & 0.229 & 0.247 & 0.208 & 0.165 & \underline{0.090} & 0.284 & FAIL & 7.09 \\
\cmidrule{2-15}
& \multirow{2}{*}{\bolden{\ours (ours)}} & \bolden{0.069} & \bolden{0.070} & \bolden{0.157} & \bolden{0.175} & \bolden{0.129} & \bolden{0.027} & 0.146 & 0.178 & \bolden{0.022} & {0.105} & \bolden{0.105} & \multirow{2}{*}{\bolden{0.11}} & \multirow{2}{*}{\bolden{2.09}} \\ 
& & \bolden{($\pm$0.002)} & \bolden{($\pm$0.002)} & \bolden{($\pm$0.009)} & \bolden{($\pm$0.009)} & \bolden{($\pm$0.002)} & \bolden{($\pm$0.000)} & ($\pm$0.001) & ($\pm$0.003) & \bolden{($\pm$0.000)} & {($\pm$0.007)} & \bolden{($\pm$0.001)} & & \\
\cmidrule{1-15}
\end{tabular}
\end{adjustbox}
\end{table}                

\begin{table}
\centering
\caption{\textbf{\ours 
generates more realistic edge weights than the baseline methods in terms of DACC and NetSimile.} 
For each dataset and each metric, we report the average over the layers generated by each method, as well as
the average (AVG) over all datasets and the average rank (A.R.), where 
the best result is in \textbf{bold}, and the second-best one is \underline{underlined}
(OOM = out of memory; FAIL = unable to output a meaningful layer).
See the online appendix~\citep{onlineSuppl} for the full results.
}
\label{tab:main_res:two}                
\begin{adjustbox}{scale=0.75}
\begin{tabular}{ c | c | *{11}{c} | c c }                    
\cmidrule{1-15}
metric & method & OF & FL & th-UB & th-MA & th-SO & sx-UB & sx-MA & sx-SO & sx-SU & co-DB & co-GE & AVG & A.R. \\
\cmidrule{1-15}
\multirow{11}{*}{DACC}
& PRD & 0.291 & 0.184 & 0.239 & 0.371 & 0.211 & 0.036 & 0.139 & 0.036 & 0.047 & 0.370 & 0.376 & 0.21 & 8.45 \\
& SCN & 0.166 & 0.417 & 0.434 & 0.457 & 0.486 & 0.587 & 0.523 & 0.520 & 0.589 & 0.160 & 0.129 & 0.41 & 9.00 \\
& SEB-R & 0.060 & 0.135 & 0.235 & 0.356 & 0.204 & 0.035 & 0.118 & 0.036 & 0.048 & 0.271 & 0.213 & 0.16 & 6.55 \\
& SEB-N & 0.246 & \textbf{0.036} & \underline{0.215} & 0.311 & 0.190 & \underline{0.018} & 0.102 & 0.026 & 0.039 & \underline{0.054} & \underline{0.046} & \underline{0.12} & \underline{3.55} \\
& PEB-R & 0.240 & 0.131 & 0.225 & 0.313 & 0.190 & 0.023 & {0.064} & \textbf{0.024} & \underline{0.018} & 0.344 & 0.341 & 0.17 & 4.64 \\
& PEB-N & 0.240 & 0.131 & 0.225 & 0.313 & 0.190 & 0.023 & \underline{0.063} & \textbf{0.024} & \textbf{0.017} & 0.344 & 0.341 & 0.17 & 4.45 \\
& STC & \textbf{0.027} & \underline{0.077} & 0.231 & \underline{0.273} & \underline{0.155} & 0.036 & 0.148 & 0.036 & 0.046 & 0.117 & 0.205 & 0.12 & 4.55 \\
& RFF & 0.187 & 0.106 & 0.279 & 0.418 & 0.399 & 0.645 & 0.468 & 0.548 & 0.595 & 0.197 & 0.120 & 0.36 & 8.36 \\
& NEB-R & 0.126 & 0.142 & 0.237 & 0.377 & 0.162 & 0.032 & 0.103 & \textbf{0.024} & 0.049 & \textbf{0.051} & \underline{0.037} & 0.12 & 4.73 \\
& NEB-N & \underline{0.028} & FAIL  & FAIL  & 0.378 & 0.195 & 0.042 & 0.173 & 0.030 & 0.031 & 0.087 & 0.175 & FAIL & 6.91 \\
\cmidrule{2-15}
& \multirow{2}{*}{\textbf{\ours (ours)}} & 0.158 & {0.080} & \textbf{0.129} & \textbf{0.262} & \textbf{0.146} & \textbf{0.016} & \textbf{0.043} & 0.037 & 0.023 & {0.227} & {0.127} & \multirow{2}{*}{\textbf{0.11}} & \multirow{2}{*}{\textbf{3.27}} \\
& & ($\pm$0.006) & {($\pm$0.007)} & \textbf{($\pm$0.004)} & \textbf{($\pm$0.001)} & \textbf{($\pm$0.001)} & \textbf{($\pm$0.000)} & \textbf{($\pm$0.000)} & ($\pm$0.000) & ($\pm$0.001) & ($\pm$0.005) & {($\pm$0.004)} & & \\
\cmidrule{1-15}				
\multirow{11}{*}{NetSimile}
& PRD & 11.99 & 16.81 & 24.71 & 23.80 & 22.38 & 15.33 & 12.49 & OOM & 15.39 & 19.83 & 19.93 & 18.27 & 7.10 \\
& SCN & 15.47 & 20.56 & 18.68 & 21.76 & 23.39 & 24.89 & 22.77 & OOM & 25.45 & \underline{12.62} & 13.32 & 19.89 & 7.30 \\
& SEB-R & 9.34  & 16.26 & 25.15 & 23.70 & 21.32 & 16.46 & 12.06 & OOM & 17.38 & 17.74 & 17.64 & 17.71 & 6.20 \\
& SEB-N & 10.67 & \underline{12.29} & 24.52 & 23.29 & 22.77 & \underline{8.51} & \textbf{8.91}  & OOM & \underline{12.85} & 16.06 & 17.34 & \underline{15.72} & \underline{4.30} \\ 
& PEB-R & 12.03 & 14.11 & 19.28 & 19.40 & \underline{17.45} & {13.04} & 12.81 & OOM & 13.65 & 18.51 & 17.55 & 15.78 & 4.90 \\
& PEB-N & 12.02 & 14.12 & 19.29 & 19.40 & 17.45 & 13.05 & 12.80 & OOM & 13.65 & 18.49 & 17.55 & 15.78 & 5.00 \\
& STC & 23.57 & 21.16 & 27.70 & 26.03 & 23.08 & 24.04 & 25.53 & OOM & 25.08 & 21.35 & 22.27 & 23.98 & 10.00 \\
& RFF & 12.26 & 13.81 & \underline{16.51} & \underline{18.10} & 18.96 & 25.61 & 20.87 & OOM & 25.42 & \textbf{14.71} & \underline{7.54}  & 17.38 & 5.30 \\
& NEB-R & 11.10 & 12.95 & 21.31 & 22.88 & 20.49 & 17.11 & 14.20 & OOM & 19.27 & \textbf{9.73}  & \underline{8.89}  & 15.79 & 4.90 \\
& NEB-N & \textbf{6.21}  & FAIL  & FAIL  & 26.69 & 22.96 & 29.80 & 27.12 & OOM & 22.49 & 19.52 & 21.79 & FAIL  & 9.20 \\
\cmidrule{2-15}
& \multirow{2}{*}{\textbf{\ours (ours)}} & \underline{8.09} & \textbf{8.92} & \textbf{9.33} & \textbf{14.85} & \textbf{13.79} & \textbf{6.03} & \underline{9.40} & \multirow{2}{*}{OOM} & \textbf{6.37} & 15.73 & 10.87 & \multirow{2}{*}{\textbf{10.34}} & \multirow{2}{*}{\textbf{1.70}} \\
& & \underline{($\pm$0.32)} & \textbf{($\pm$0.17)} & \textbf{($\pm$0.22)} & \textbf{($\pm$0.07)} & \textbf{($\pm$0.10)} & \textbf{($\pm$0.12)} & \underline{($\pm$0.04)} & & \textbf{($\pm$0.14)} & ($\pm$0.24) & ($\pm$0.04) & & \\
\cmidrule{1-15}
\end{tabular}
\end{adjustbox}
\end{table}                

\endgroup

\subsection{Evaluation methods}\label{subsec:eval_methods}

We would re-emphasize that we aim to generate edge weights that are realistic in a \textit{macroscopic} scope, and thus the evaluation should also be in a macroscopic scope.
We report the following metrics including several graph statistics that have been widely used for evaluating graph generators~\citep{leskovec2005realistic, shuai2013pattern, cao2015grarep, heath2011generating} {to compare, for each $2\leq i\leq 5$, the layer-$i$ produced by each method and the original one}:
(1) KS statistic for number-of-CN distributions \textbf{(KSCN)},
(2) KS statistic for node-degree distributions \textbf{(KSND)},
(3) difference in average clustering coefficients \textbf{(DACC)}, and
(4) a graph distance measure computed by \textbf{NetSimile}~\citep{berlingerio2012netsimile}.\footnote{Given a graph, NetSimile uses seven node-level structural features to generate a characteristic vector for the graph after feature aggregation over the nodes.
	Notably, we do not need to solve the node-correspondence problem for NetSimile and the measure is size-invariant~\citep{berlingerio2012netsimile}.
	NetSimile runs out of memory on \textit{sx-SO} and the corresponding results are unavailable.}
{The intuition is that if two weighted graphs have similar layers with the same layer index, then the two weighted graphs are similar too.}
Note that the evaluation focuses on the first four layers since, for $i>5$, the ground-truth layer-$i$ is too small or too sparse in some datasets.
{See Appendix~\ref{app:res_motifs} for more analysis using graph motifs.}

\subsection{Results}
First, we show how the methods perform on each dataset.
In Tables~\ref{tab:main_res} and \ref{tab:main_res:two}, for each dataset, each metric, and each method, we report the average value over all the generated layers.
{For \ours, the mean value and the standard deviation over three trials of each setting are reported.}
Overall, \ours has the best average value and the highest average rank among all the methods w.r.t each metric.
Specifically, \ours achieves an average rank of $2.18$, $2.09$, $3.27$, and $1.70$ (there are $11$ methods in total), w.r.t the metric KSCN, KSNC, DACC, and NetSimile, respectively.
{Notably, although supervised with some ground-truth edge weights, the supervised baseline methods do not show clear superiority over the unsupervised ones. In our understanding, this is because the edge-weight classes (i.e., layers) are highly imbalanced (i.e., the numbers of edges with different weights vary a lot), while the ground-truth numbers of edges in each edge-weight class are provided to the unsupervised baseline methods including SEB are helpful.}
{Also, the methods using sophisticated embeddings are sometimes even worse than SCN which assigns edge weights using a simple heuristic based on the number of CNs. In our understanding, this is because local information is important (according to our observations), while the embedding-based methods focus on higher-order information which can be confusing and harmful to the prediction results.}

In Figure~\ref{fig:evaluation-dist}, for the two datasets \textit{th-UB} and \textit{sx-MA}, and for $i \in \setpar{2, 3, 4}$, we report the node-degree and edge-CN distributions of the layer-$i$ generated by each method, where the layers generated by \ours have the two distributions closest to the original ones.
{For each baseline method using embeddings, between the two variants using RandNE and node2vec, we report the results of the variant that performs better.}

We also study how the performance {varies across generated layers.}
In Figure~\ref{fig:res_by_layer}, for each method, each metric, and each layer index $i$, we report the average over all datasets.
{Again, for each baseline method using embeddings, between the two variants using RandNE and node2vec, we report the results of the variant that performs better.}
Overall, \ours performs best for each layer.

\section{{Conclusion and future directions}\label{sec:disc_concl}}

In this work, we explored the relations between edge weights and topology in real-world graphs.
We proposed a new concept called layers (Definition~\ref{def:layers})
with several related concepts (Definition~\ref{def:SEs_and_FSEs}),
and observed several pervasive patterns (Observations~\ref{obs:adj_and_str}-\ref{obs:fose_power_law}).
We also proposed \ours (Algorithm~\ref{alg:weight-assign}), a weight-assignment algorithm with only two parameters, based on the formulation of our observations (Properties~\eqref{eq:foses_linear}-\eqref{eq:foses_powerlaw} and Equations~\eqref{eq:p0-se}-\eqref{eq:solve_se}).
In our experiments on eleven real-world graphs, we showed that \ours generates more realistic edge weights than the baseline methods, {including those requiring much more prior knowledge and parameters (Tables~\ref{tab:main_res}-\ref{tab:main_res:two} and Figures \ref{fig:evaluation-dist}-\ref{fig:res_by_layer}).}

{The observations in this work are macroscopic and the proposed model is based on those observations. Although macroscopic patterns are relatively more robust to outliers compared to microscopic ones, it is still possible that outliers can impair the performance of our method. We plan to explore and examine how outliers can affect the performance of our method and the corresponding solutions.}
{We studied weighted graphs with positive-integer edge weights, and we plan to extend the study on the interplay between edge weights and topology to real-valued edge weights.}
We studied the average behavior of each group of edges or pairs with the same number of common neighbors, and we plan to further explore the patterns within each group.
We also plan to study how the two parameters of \ours affect the output.



\bibliographystyle{plainnat}
\bibliography{ref}
\clearpage

\appendix

\begin{table*}[t!]
	\begin{center}
		\caption{\textbf{The numbers of common neighbors are simple yet indicative.} Additional results: we report the area under the ROC curve (AUC) of the logistic regression fit to the sequences of each quantity and 
			the binary indicators of repetition.}
		\label{tab:why_cn_auc}
		\resizebox{0.968\linewidth}{!}{%
			\begin{tabular}{ l *{14}{|c} }
				\toprule
				\textbf{dataset} & NC & SA & JC & HP & HD & SI & LI & AA & RA & PA & FM & DL & EC & LP \\
				\midrule
				OF  & 0.704 & 0.636 & 0.641 & 0.508 & 0.637 & 0.641 & 0.632 & 0.707 & 0.723 & 0.705 & 0.571 & 0.692 & 0.666 & 0.704 \\               
				FL      & 0.676 & 0.664 & 0.657 & 0.619 & 0.647 & 0.657 & 0.479 & 0.698 & 0.725 & 0.629 & 0.546 & 0.606 & 0.610 & 0.674 \\              
				th-UB   & 0.890 & 0.689 & 0.702 & 0.590 & 0.704 & 0.702 & 0.499 & 0.879 & 0.832 & 0.122 & 0.805 & 0.781 & 0.872 & 0.895 \\               
				th-MA   & 0.843 & 0.751 & 0.701 & 0.587 & 0.689 & 0.701 & 0.662 & 0.845 & 0.823 & 0.834 & 0.800 & 0.741 & 0.801 & 0.844 \\               
				th-SO   & 0.831 & 0.717 & 0.687 & 0.646 & 0.676 & 0.687 & 0.580 & 0.837 & 0.808 & 0.199 & 0.796 & 0.756 & 0.768 & 0.827 \\
				sx-UB   & 0.586 & 0.579 & 0.576 & 0.565 & 0.575 & 0.576 & 0.446 & 0.585 & 0.581 & 0.592 & 0.582 & 0.576 & 0.577 & 0.598 \\
				sx-MA   & 0.635 & 0.631 & 0.613 & 0.574 & 0.606 & 0.613 & 0.506 & 0.639 & 0.645 & 0.623 & 0.605 & 0.595 & 0.611 & 0.634 \\
				sx-SO   & 0.562 & 0.561 & 0.551 & 0.560 & 0.549 & 0.551 & 0.463 & 0.563 & 0.566 & 0.553 & 0.555 & 0.552 & 0.538 & 0.563 \\
				sx-SU   & 0.588 & 0.576 & 0.570 & 0.570 & 0.569 & 0.570 & 0.455 & 0.587 & 0.583 & 0.593 & 0.585 & 0.580 & 0.578 & 0.596 \\
				co-DB   & 0.624 & 0.529 & 0.523 & 0.530 & 0.521 & 0.523 & 0.587 & 0.617 & 0.586 & 0.360 & 0.624 & 0.627 & 0.601 & 0.630 \\
				co-GE   & 0.684 & 0.527 & 0.512 & 0.551 & 0.505 & 0.512 & 0.616 & 0.687 & 0.633 & 0.679 & 0.654 & 0.658 & 0.651 & 0.686\\
				\midrule
				avg. & 0.693 & 0.624 & 0.612 & 0.573 & 0.607 & 0.612 & 0.539 & 0.695 & 0.682 & 0.535 & 0.648 & 0.651 & 0.661 & 0.696 \\
				avg. rank & 3.00 & 8.18 & 9.18 & 11.18 & 11.18 & 9.18 & 12.45 & 2.45 & 4.18 & 7.27 & 7.36 & 7.73 & 7.73 & 2.27 \\
				\bottomrule 
			\end{tabular}
		}
	\end{center}
\end{table*}

\section{General edge weights}\label{app:general_edge_weights}
As mentioned in Remark~\ref{rem:general_edge_weights}, all the observations in this work are based on layer structures, and the edge weights are limited to integer ones representing the number of occurrences.
In order to examine the generality of our observations and the proposed algorithm,
we examined several blockchain transaction datasets~\citep{kilicc2022analyzing, kilicc2022parallel}, where each edge weight is a real number representing the amount of the corresponding transaction.
On the transaction datasets mentioned above, the correlations between the edge weights and the number of common neighbors are very low (consistently less than $0.1$).
Therefore, our observations and the proposed algorithm can not be directly applied to these datasets, which is a limitation of this work.

For general weighted graphs, we may make use of the known observation that real-world weighted graphs often have a heavy-tailed edge-weight distribution~\citep{barrat2004modeling,kumar2020retrieving,starnini2017robust}.
We leave as a potential future direction the exploration of the relation between edge weights and the number of common neighbors on weighted graphs with more general edge weights, as well as weighted graphs with edge weights having different real-world meanings (other than the number of occurrences).

\section{The indicativeness of the number of common neighbors}\label{app:why_cn_auc}
In Section~\ref{subsec:why_cns}, we have compared the indicativeness of the number of common neighbors with some baseline quantities measured by the point-biserial correlation coefficients (which are mathematically equivalent to the Pearson correlation coefficients).
We would like to also examine this relation using other metrics, e.g., the area under the ROC curve (AUC).

In Table~\ref{tab:why_cn_auc}, we report the additional results measured by the AUC.
Specifically, for each quantity, we first perform logistic regression between the sequence of the quantity and the binary indicators of repetition, then we compute the AUC of the output prediction.
Although there are some baseline quantities achieving marginally higher AUC values, the number of common neighbors is still one of the most promising quantities, especially when we consider its simplicity.

\begin{table}[t!]
	\begin{center}
		\caption{\textbf{The saturation point of the fractions of adjacent pairs does not change much when we change the threshold in the definition from $100\%$ to $99\%$.} For each layer of each dataset, we report the saturation point of the fractions of adjacent pairs with threshold $100\%$ on the left and the saturation point with threshold $99\%$ on the right.}
		\label{tab:saturation_pts_99_100}
		\resizebox{0.6\linewidth}{!}{%
			\begin{tabular}{ lccccc }
				\toprule
				\textbf{dataset} & layer-$1$ & layer-$2$ & layer-$3$ & layer-$4$ & layer-$5$ \\
				\midrule
				OF          & 241/234 & 190/190 & 157/157 & 134/134 & 119/119 \\ 
				\midrule
				FL          & 64/64 & 31/31 & 17/17 & - & - \\
				\midrule
				th-UB       & 73/73 & 30/30 & 19/19 & 18/18 & 15/15 \\
				th-MA       & 372/372 & 145/145 & 114/114 & 84/84 & 63/63 \\
				th-SO       & 685/685 & 208/208 & 134/134 & 97/97 & 74/74 \\
				\midrule
				sx-UB       & 152/152 & 63/63 & 48/48 & 36/36 & 31/31 \\
				sx-MA       & 185/185 & 113/113 & 75/75 & 60/60 & 51/51 \\
				sx-SO       & 886/886 & 407/407 & 221/221 & 169/169 & 120/120 \\
				sx-SU       & 202/202 & 96/96 & 63/63 & 48/48 & 36/36 \\
				\midrule
				co-DB       & 83/83 & 36/20 & 22/17 & 20/17 & 16/14 \\
				co-GE       & 74/74 & 52/46 & 34/32 & 28/27 & 24/22 \\
				\bottomrule %
			\end{tabular}
		}
	\end{center}
\end{table}

\section{Saturation points}\label{app:saturation_pts}
In Definition~\ref{def:saturation_pts}, we have defined the saturation point of the fraction of weighty edges (adjacent pairs) as the number $c$ of common neighbors such that all (i.e., $100\%$) the edges (pairs) sharing $c$ common neighbors are weighty (adjacent).
Theoretically, this definition can be less robust since a single edge (pair) that is not weighty (adjacent) can affect the whole group of edges (pairs) sharing the same number of common neighbors.
We have chosen to use the current definition for simplicity and clarity,
while it is possible to make the concept more robust by using a less ``absolute'' threshold, e.g., ``$99\%$''.

In Table~\ref{tab:saturation_pts_99_100}, we show how the saturation points in the datasets used in our experiments change when we change the ``$100\%$'' in the definition to ``$99\%$''.
In the real-world datasets that we use, such a change makes no big difference.
Therefore, we claim that such a simple and clear definition is expected to work well in practice,
but practitioners may take the issue discussed above into consideration for better robustness.

\begin{table}
	\centering
	\caption{\textbf{The parameter sensitivity results of \ours.} We report the performance of \ours with four different parameter settings on the eleven real-world datasets used in our experiments, with the average rank of \ours among all the baseline methods and \ours for each parameter.}
	\label{tab:params_res}         
	\begin{adjustbox}{scale=0.8}
		\begin{tabular}{c|c|c|c|c|c|c|c|c}
			\toprule
			metric & \multicolumn{2}{c|}{KSCN} & \multicolumn{2}{c|}{KSND} & \multicolumn{2}{c|}{DACC} & \multicolumn{2}{c}{NetSimile} \\
			\midrule
			dataset & {(0.98, 1.02)} & {(0.7, 1.3)} & {(0.98, 1.02)} & {(0.7, 1.3)} & {(0.98, 1.02)} & {(0.7, 1.3)} & {(0.98, 1.02)} & {(0.7, 1.3)} \\
			\midrule
			OF    & {0.391$\pm$0.002} & {0.074$\pm$0.003} & {0.254$\pm$0.000} & {0.069$\pm$0.002} & {0.057$\pm$0.001} & {0.158$\pm$0.006} & {12.754$\pm$0.038} & {8.087$\pm$0.320} \\
			FL    & {0.425$\pm$0.001} & {0.241$\pm$0.017} & {0.159$\pm$0.002} & {0.070$\pm$0.002} & {0.141$\pm$0.003} & {0.080$\pm$0.007} & {16.563$\pm$0.063} & {8.918$\pm$0.171} \\
			th-UB & {0.482$\pm$0.002} & {0.168$\pm$0.015} & {0.188$\pm$0.002} & {0.157$\pm$0.009} & {0.187$\pm$0.001} & {0.129$\pm$0.004} & {18.661$\pm$0.031} & {9.327$\pm$0.224} \\
			th-MA & {0.239$\pm$0.003} & {0.186$\pm$0.006} & {0.071$\pm$0.001} & {0.175$\pm$0.000} & {0.151$\pm$0.003} & {0.262$\pm$0.001} & {15.680$\pm$0.028} & {14.854$\pm$0.074} \\
			th-SO & {0.544$\pm$0.001} & {0.126$\pm$0.008} & {0.109$\pm$0.002} & {0.129$\pm$0.002} & {0.189$\pm$0.000} & {0.146$\pm$0.001} & {18.698$\pm$0.069} & {13.786$\pm$0.098} \\
			sx-UB & {0.076$\pm$0.003} & {0.416$\pm$0.003} & {0.027$\pm$0.000} & {0.118$\pm$0.010} & {0.016$\pm$0.000} & {0.071$\pm$0.004} & {6.026$\pm$0.115} & {17.414$\pm$0.094} \\
			sx-MA & {0.072$\pm$0.001} & {0.176$\pm$0.009} & {0.146$\pm$0.001} & {0.098$\pm$0.003} & {0.043$\pm$0.000} & {0.059$\pm$0.000} & {9.397$\pm$0.041} & {11.684$\pm$0.084} \\
			sx-SO & {0.239$\pm$0.001} & {0.315$\pm$0.004} & {0.178$\pm$0.003} & {0.202$\pm$0.001} & {0.037$\pm$0.000} & {0.038$\pm$0.001} & {N/A} & {N/A} \\
			sx-SU & {0.103$\pm$0.002} & {0.376$\pm$0.000} & {0.022$\pm$0.000} & {0.116$\pm$0.005} & {0.023$\pm$0.001} & {0.057$\pm$0.003} & {6.366$\pm$0.144} & {17.237$\pm$0.112} \\
			co-DB & {0.329$\pm$0.013} & {N/A} & {0.105$\pm$0.007} & {N/A} & {0.227$\pm$0.005} & {N/A} & {15.731$\pm$0.236} & {N/A} \\
			co-GE & {0.181$\pm$0.008} & {0.431$\pm$0.009} & {0.105$\pm$0.001} & {0.375$\pm$0.005} & {0.127$\pm$0.004} & {0.341$\pm$0.001} & {10.871$\pm$0.038} & {17.722$\pm$0.056} \\
			\midrule
			AVG   & 0.280  & 0.287 & 0.124 & 0.178 & 0.109 & 0.157 & 13.075 & 14.306 \\
			\midrule
			A.R.  & 3.09  & 4.45  & 3.09  & 4.64  & 3.64  & 5.27  & 3.30   & 3.90 \\
			\bottomrule
			\multicolumn{9}{c}{}\\
			\toprule
			metric & \multicolumn{2}{c|}{KSCN} & \multicolumn{2}{c|}{KSND} & \multicolumn{2}{c|}{DACC} & \multicolumn{2}{c}{NetSimile} \\
			\midrule
			dataset & {(0.9, 1.1)} & {(0.5, 1.5)} & {(0.9, 1.1)} & {(0.5, 1.5)} & {(0.9, 1.1)} & {(0.5, 1.5)} & {(0.9, 1.1)} & {(0.5, 1.5)} \\
			\midrule
			OF & 0.247$\pm$0.005 & 0.105$\pm$0.004 & 0.177$\pm$0.003 & 0.144$\pm$0.003 & 0.030$\pm$0.001 & 0.185$\pm$0.001 & 10.471$\pm$0.404 & 9.617$\pm$0.075 \\ 
			FL & 0.212$\pm$0.012 & 0.323$\pm$0.007 & 0.056$\pm$0.005 & 0.107$\pm$0.010 & 0.054$\pm$0.010 & 0.086$\pm$0.009 & 12.116$\pm$0.062 & 11.626$\pm$0.231 \\ 
			th-UB & 0.375$\pm$0.002 & 0.221$\pm$0.007 & 0.269$\pm$0.005 & 0.093$\pm$0.007 & 0.212$\pm$0.001 & 0.043$\pm$0.003 & 17.422$\pm$0.089 & 9.637$\pm$0.028 \\ 
			th-MA & 0.350$\pm$0.002 & 0.222$\pm$0.011 & 0.184$\pm$0.001 & 0.099$\pm$0.001 & 0.305$\pm$0.001 & 0.158$\pm$0.001 & 15.438$\pm$0.140 & 14.418$\pm$0.092 \\ 
			th-SO & 0.393$\pm$0.003 & 0.152$\pm$0.004 & 0.206$\pm$0.001 & 0.082$\pm$0.002 & 0.199$\pm$0.000 & 0.071$\pm$0.002 & 17.851$\pm$0.090 & 13.295$\pm$0.049 \\ 
			sx-UB & 0.141$\pm$0.003 & 0.526$\pm$0.008 & 0.108$\pm$0.001 & 0.209$\pm$0.011 & 0.016$\pm$0.000 & 0.176$\pm$0.011 & 13.209$\pm$0.099 & 19.761$\pm$0.501 \\ 
			sx-MA & 0.153$\pm$0.000 & 0.266$\pm$0.004 & 0.071$\pm$0.002 & 0.106$\pm$0.004 & 0.090$\pm$0.000 & 0.072$\pm$0.002 & 6.167$\pm$0.147 & 14.421$\pm$0.134 \\ 
			sx-SO & 0.143$\pm$0.002 & 0.428$\pm$0.001 & 0.297$\pm$0.001 & 0.143$\pm$0.001 & 0.030$\pm$0.001 & 0.098$\pm$0.003 & N/A & N/A \\ 
			sx-SU & 0.135$\pm$0.002 & 0.492$\pm$0.003 & 0.137$\pm$0.002 & 0.180$\pm$0.004 & 0.027$\pm$0.000 & 0.158$\pm$0.011 & 13.521$\pm$0.123 & 19.791$\pm$0.260 \\ 
			co-DB & N/A & N/A & N/A & N/A & N/A & N/A & N/A & N/A \\ 
			co-GE & 0.486$\pm$0.001 & 0.464$\pm$0.003 & 0.386$\pm$0.000 & 0.288$\pm$0.005 & 0.344$\pm$0.000 & 0.299$\pm$0.004 & 17.732$\pm$0.023 & 17.244$\pm$0.080 \\ 
			\midrule
			AVG & 0.299 & 0.336 & 0.214 & 0.162 & 0.154 & 0.151 & 14.775 & 15.211 \\ 
			\midrule
			A.R. & 3.18 & 4.45 & 5.73 & 4.82 & 4.36 & 5.27 & 3.60 & 4.30 \\ 
			\bottomrule
		\end{tabular}%
	\end{adjustbox}
\end{table}                

\section{Parameters}\label{app:choose_params}
As discussed in Section~\ref{sec:eval}, \ours uses only two parameters ($a$ and $k$), and for the eleven real-world datasets used in our experiments, we only use \textit{two} parameter settings $(a, k) \in \setpar{(0.98, 1.02), (0.7, 1.3)}$ without fitting to the ground-truth edge weight.

In Table~\ref{tab:params_res}, we report the results on the \textit{parameter sensitivity} of \ours,
where we provide the performance of \ours with four different parameter settings (including the two considered parameter settings $(a, k) \in$ $\{(0.98, 1.02),$ $(0.7, 1.3)\}$).
In the table, for each parameter setting, we also report the average rank of \ours (as in Tables~\ref{tab:main_res} and \ref{tab:main_res:two}) if we use the setting for \ours.

One observed limitation of \ours is its considerable \textit{parameter sensitivity}.
That is, using \ours with different parameters can give fairly different performances.
Therefore, it is important to find a well-working group of parameters when using \ours.

First, we would like to emphasize that finding the best parameters is not a trivial problem.
This is because, for a given topology, multiple plausible groups of edge weights are possible (consider, e.g., multiple snapshots of an evolving weighted graph).

Fortunately, \ours works well without extensive parameter searching (specifically, only with two parameter settings) in the form of $(1 - x, 1 + x)$,
and we have observed that for datasets within the same domain, we can use the same parameter setting.
Notably, even simply using the same parameters $(a, k) = (0.98, 1.02)$, \ours still achieves the best average value for each metric (compare the values in Tables~\ref{tab:main_res}, \ref{tab:main_res:two}, and \ref{tab:params_res}).

Another interesting and insightful observation is that the well-performing parameters of the datasets are seemingly correlated to the correlation between the numbers of common neighbors and the repetition of edges.
Based on the point-biserial correlation coefficients between the sequences of the numbers of common neighbors and the binary indicators of repetition in Table~\ref{tab:pearson-rep}, we can see that for all the datasets with a high coefficient (specifically, larger than $0.30$), the parameter setting $(a, k) = (0.7, 1.3)$ performs better than $(a, k) = (0.98, 1.02)$, while for all the dataset with a relatively low coefficient, the situation is the opposite.
Although the repetition of edges cannot be directly obtained merely from the topology, this provides insights into the reasons why \ours with different parameters shows different performance on different datasets.
We leave finding optimal parameters of \ours for a given topology as a future direction.

\begin{table}[t!]
	\begin{center}
		\caption{\textbf{The predicted edge weights output by \ours can enhance the community detection performance of the Louvain method.}
			For seven real-world datasets, we report the community detection performance of the Louvain method on the original unweighted graph and the weighted counterpart with edge weighted output by \ours.
			The performance is measured by adjusted rand index (ARI) and normalized mutual information (NMI). The better results are marked in bold.}
		\label{tab:comm_detect}
		\resizebox{0.9\linewidth}{!}{%
			\begin{tabular}{ l|cc|cc }
				\toprule
				\textbf{dataset}    & ARI (unweighted) & ARI (\ours) & NMI (unweighted)  & NMI (\ours)    \\
				\midrule				
				cora                
				& $0.2481 \pm 0.0189$   & $\mathbf{0.2507 \pm 0.0151}$
				& $0.4546 \pm 0.0069$   & $\mathbf{0.4570 \pm 0.0055}$ \\
				citeseer
				& $0.0937 \pm 0.0069$   & $\mathbf{0.0950 \pm 0.0059}$
				& $0.3287 \pm 0.0027$   & $\mathbf{0.3292 \pm 0.0020}$ \\
				pubmed
				& $0.0946 \pm 0.0034$   & $\mathbf{0.0948 \pm 0.0083}$   
				& $0.1774 \pm 0.0033$   & $\mathbf{0.1779 \pm 0.0035}$ \\                             
				computer
				& $0.3147 \pm 0.0119$   & $\mathbf{0.3201 \pm 0.0199}$   
				& $0.5411 \pm 0.0083$   & $\mathbf{0.5459 \pm 0.0056}$ \\
				photo
				& $0.5696 \pm 0.0301$   & $\mathbf{0.5796 \pm 0.0074}$   
				& $0.6673 \pm 0.0165$   & $\mathbf{0.6722 \pm 0.0069}$ \\
				cornell           
				& $0.0230 \pm 0.0006$   & $\mathbf{0.0274 \pm 0.0011}$   
				& $0.0956 \pm 0.0019$   & $\mathbf{0.1014 \pm 0.0030}$ \\
				texas
				& $0.0513 \pm 0.0025$   & $\mathbf{0.0758 \pm 0.0007}$   
				& $0.0698 \pm 0.0014$   & $\mathbf{0.0835 \pm 0.0030}$ \\
				wisconsin
				& $0.0230 \pm 0.0039$   & $\mathbf{0.0290 \pm 0.0045}$   
				& $0.0911 \pm 0.0057$   & $\mathbf{0.0977 \pm 0.0047}$ \\
				\bottomrule %
			\end{tabular}
		}
	\end{center}
\end{table}

\section{An application: community detection}\label{app:application_com_det}
As mentioned in the introduction (Section~\ref{sec:intro}), assigning realistic edge weights to a given topology is an important and practical problem, and the predicted edge weights output by \ours can be used in many practical applications.
Here, we showcase one possible application, where we use the predicted edge weights to enhance the performance of community detection algorithms on graphs.

We use seven datasets with ground-truth clusters but without groud-truth edge weights, including
(1-2) citation networks \textit{cora} and \textit{citeseer}~\citep{sen2008collective},
(3-5) website networks \textit{cornell}, \textit{texas}, and \textit{wisconsin}~\citep{pei2020geom}, and
(6-7) co-purchasing networks \textit{computer} and \textit{photo}~\citep{shchur2018pitfalls}.\footnote{More details of these datasets can be found at \url{https://pytorch-geometric.readthedocs.io/en/latest/notes/data_cheatsheet.html}~\citep{fey2019fast}.}

For each dataset, we compare the performance of the Louvain method~\citep{blondel2008fast} on the original unweighted graph and on the weighted counterpart with edge weights predicted by \ours.
For simplicity, we only predict a single layer, i.e., layer-$2$. 
The parameter settings $(a, k)$ used on the datasets are: 
$(0.95, 1.05)$ for \textit{cora};
$(0.99, 1.01)$ for \textit{citeseer}, \textit{pubmed}, \textit{computer}, and \textit{photo}; 
and $(0.9, 1.1)$ for \textit{cornell}, \textit{texas}, and \textit{wisconsin}.
We then measure the performance w.r.t the adjusted rand index (ARI) and the normalized mutual information (NMI).

See Table~\ref{tab:comm_detect} for the detailed results,
where we can observe that the predicted edge weights consistently improve the community detection performance of the Louvain method, and we leave further exploration of this application (e.g., the reasons behind this improvement) as a future direction.
In our understanding, the edge weights output by \ours reinforce the local structures and thus enhance the performance.
This has been shown in some works~\citep{satuluri2011local,benson2016higher,tsourakakis2017scalable}, where triangle counts are shown to be a helpful feature for community detection.
Even on weighted graphs with ground-truth edge weights, one can still use PEAR to obtain another group of edge weights and use the edge weights predicted by PEAR as additional features.

\begin{table}[t!]
	\begin{center}
		\caption{\textbf{Evaluation using graph motifs.} For each dataset, and for each $2 \leq i \leq 5$, we compare the layer-$i$ output by each method, by counting the number of different $3$-motifs (induced subgraphs of size $3$, up to graph isomorphism) and comparing the distributions (i.e., vectors of size $3$). 
			The final performance is measured by the L1 difference (which is equivalent to the total variance) between the distribution in the original graph and that in the output one.
			The better results (i.e., smaller errors) are marked in bold.
			Note that SEB* uses the ground-truth number of edges in each layer.
		}
		\label{tab:res_motifs}
		\resizebox{0.8\linewidth}{!}{%
			\begin{tabular}{ l|c|c|c|c|c|c|c|c }
				\toprule
				\textbf{dataset} & 
				\multicolumn{2}{c|}{layer-$2$} & 
				\multicolumn{2}{c|}{layer-$3$} & 
				\multicolumn{2}{c|}{layer-$4$} & 
				\multicolumn{2}{c}{layer-$5$} \\            
				\midrule
				\textbf{method}
				& SEB* & \ours 
				& SEB* & \ours
				& SEB* & \ours
				& SEB* & \ours \\
				\midrule
				OF          
				& {6.26E-2} & \bolden{1.05E-2}
				& {1.11E-1} & \bolden{1.30E-2}
				& {1.41E-1} & \bolden{4.44E-2}
				& {1.57E-1} & \bolden{6.40E-2}
				\\
				\midrule
				FL          
				& {1.57E-2} & \bolden{3.38E-3}
				& {2.51E-2} & \bolden{6.27E-3}
				& {3.46E-2} & \bolden{1.37E-2}
				& {4.59E-2} & \bolden{2.54E-2}
				\\
				\midrule
				th-UB       
				& {2.89E-2} & \bolden{1.82E-2}
				& {6.88E-2} & \bolden{3.09E-2}
				& {1.14E-1} & \bolden{4.00E-2}
				& {1.55E-1} & \bolden{3.20E-2}
				\\
				th-MA       
				& {3.09E-2} & \bolden{2.12E-2}
				& {5.47E-2} & \bolden{3.19E-2}
				& {7.30E-2} & \bolden{3.46E-2}
				& {8.71E-2} & \bolden{2.90E-2}
				\\
				th-SO       
				& {2.48E-3} & \bolden{7.30E-4}
				& {4.75E-3} & \bolden{3.84E-3}
				& \bolden{6.61E-3} & {1.55E-2}
				& \bolden{8.07E-3} & {3.21E-2}
				\\
				\midrule
				sx-UB       
				& {2.09E-3} & \bolden{1.10E-4}
				& {3.24E-3} & \bolden{2.85E-3}
				& \bolden{3.74E-3} & {6.76E-3}
				& \bolden{3.71E-3} & {1.15E-2}
				\\
				sx-MA       
				& {6.06E-3} & \bolden{2.13E-3}
				& {1.11E-2} & \bolden{5.41E-3}
				& {1.54E-2} & \bolden{7.49E-3}
				& {1.95E-2} & \bolden{9.20E-3}
				\\
				sx-SO       
				& {2.52E-4} & \bolden{4.99E-5}
				& \bolden{3.73E-4} & {9.10E-4}
				& \bolden{4.49E-4} & {2.50E-3}
				& \bolden{5.39E-4} & {4.56E-3}
				\\
				sx-SU       
				& {2.81E-3} & \bolden{1.19E-3}
				& {5.04E-3} & \bolden{5.65E-4}
				& {7.38E-3} & \bolden{3.01E-3}
				& {1.00E-2} & \bolden{5.38E-3}
				\\
				\midrule
				co-DB       
				& {4.56E-5} & \bolden{7.72E-6}
				& {5.59E-5} & \bolden{1.19E-5}
				& {6.15E-5} & \bolden{1.69E-5}
				& {6.59E-5} & \bolden{2.23E-5}
				\\
				co-GE       
				& {1.26E-4} & \bolden{4.84E-5}
				& {1.65E-4} & \bolden{8.43E-5}
				& {1.92E-4} & \bolden{1.13E-4}
				& {2.16E-4} & \bolden{1.40E-4}
				\\
				\bottomrule %
			\end{tabular}
		}
	\end{center}
\end{table}

\section{Evaluation using graph motifs}\label{app:res_motifs}
In this section, we provide the additional results where we evaluate predicted edge weights by analyzing graph motifs in each layer.
Specifically, for each dataset, and for each $2 \leq i \leq 5$, we compare the layer-$i$ output by each method, by counting the number of different $3$-motifs (induced subgraphs of size $3$, up to graph isomorphism) and comparing the distributions.
There are three kinds of (nonempty) $3$-motifs (up to graph isomorphism):
(1) a single edge and an isolated node,
(2) a wedge (i.e., an open triangle), and
(3) a (closed) triangle.
In each output layer, we iterate all the $3$-subgraphs and count the $3$-motifs.
Then we use their ratios to obtain a sum-one vector of size $3$ for each case.
The final performance is measured by 
the L1 difference (which is equivalent to the total variance~\citep{garner1956relation})
between the distribution in the original graph and that in the output one.
This comparison is related to the difference in average clustering coefficients in that both comparisons involve triangles, but the evaluation using graph motifs provides a different perspective, focusing more on the local patterns.
For brevity, we only report the results of \ours and SEB-N, which perform consistently better than the other method.
See Table~\ref{tab:res_motifs} for the detailed results,
where in most cases \ours shows better performance.

\color{black}

\end{document}